\documentclass[a4paper,11pt,twoside]{article}
\pdfoutput=1

\usepackage{microtype}
\usepackage[charter]{mathdesign} 

\usepackage{amsmath}
\usepackage{eucal}

\usepackage[colorlinks=true]{hyperref}

\usepackage{color}
\definecolor{hypercolor}{rgb}{0,0.2,0.7}
\hypersetup{
  linkcolor=hypercolor,
  urlcolor=hypercolor,
  citecolor=hypercolor,
  pdfauthor={Claudio Dappiaggi \& Daniel Siemssen},
  pdfcreator=LaTeX,
  pdfproducer=pdfTeX
}

\usepackage{amsthm}

\newtheorem{theorem}{Theorem}[section]
\newtheorem{proposition}[theorem]{Proposition}
\newtheorem{lemma}[theorem]{Lemma}
\newtheorem{corollary}[theorem]{Corollary}
\newtheorem{definition}[theorem]{Definition}

\numberwithin{equation}{section}

\newlength{\dinwidth}
\newlength{\dinmargin}

\newcommand{\email}[1]{E-mail: \href{mailto:#1}{#1}}

\DeclareMathOperator{\supp}{supp}

\begin{document}


\bigskip
\LARGE
\noindent
\textbf{Hadamard States for the Vector Potential on Asymptotically Flat Spacetimes}

\bigskip\bigskip

\large
\noindent
\textbf{Claudio Dappiaggi\rlap{$^1$},\, Daniel Siemssen$^2$}

\medskip

\small
\hangindent 1em\noindent\hbox to 1em{\smash{$^1$}}\unskip
Dipartimento di Fisica, Universit\`a di Pavia \& INFN, sezione di Pavia -- Via Bassi 6, I-27100 Pavia, Italy. \email{claudio.dappiaggi@unipv.it}

\hangindent 1em\noindent\hbox to 1em{\smash{$^2$}}\unskip
Dipartimento di Matematica, Universit\`a di Genova, Via Dodecaneso 35, I-16146 Genova, Italy. \email{siemssen@dima.unige.it}

\bigskip

\paragraph{Abstract.}
We develop a quantization scheme for the vector potential on globally hyperbolic spacetimes which realizes it as a locally covariant conformal quantum field theory.
This result allows us to employ on a large class of backgrounds, which are asymptotically flat at null infinity, a bulk-to-boundary correspondence procedure in order to identify for the underlying field algebra a distinguished ground state which is of Hadamard form.

\bigskip



\section{Introduction}

A key feature of all free quantum field theories on Minkowski spacetime is covariance under the action of the underlying isometries, namely the Poincar\'e group.
At a quantum level, Poincar\'e invariance further leads to the identification of a unique quasi-free, pure, ground state \cite{haag:1996zi}.
Yet, as soon as we consider a non-trivial but fixed background, all these peculiarities disappear -- no matter which free quantum field theory we consider -- since the underlying isometry group is, in general, much smaller than the Poincar\'e group.
As a result of this, it is not only impossible to select a preferred ground state on a curved background but one also has to be more careful when deciding whether a certain state is physically sensible or not.
Nevertheless, in the past few years it has become universally accepted that an answer to this question lies in the requirement that a physical state should be of Hadamard form, {\it i.e.}, it should satisfy a specific condition on the singular structure of its two-point function \cite{radzikowski:1996ug,radzikowski:1996ul}.
The main disadvantage of this paradigm is that, while it is extremely appealing from a mere mathematical point of view, it is rather hard to come up with a scheme to construct any such state, unless the underlying background is rather special, such as for example a globally hyperbolic and static spacetime \cite{sahlmann:2000kl}.

A potential way to circumvent this problem has been put forth in the past few years and it calls for employing a so-called bulk-to-boundary reconstruction technique.
According to this method, if a spacetime possesses a distinguished codimension $1$ null submanifold, it is possible to associate to the latter an intrinsic $*$-algebra on which one can try to encode the information of the bulk theory.
More precisely, one has to construct an injective $*$-homomorphism between the algebra of observables of the chosen bulk free field theory and the boundary theory.
This has a twofold advantage: Every state for the boundary algebra induces automatically via the injection map a bulk counterpart whose properties can be studied and it is much easier to explicitly construct states on a three-dimensional null manifold than to construct them directly in the bulk spacetime.
As a matter of fact, this idea has already been successfully applied to construct distinguished Hadamard states for the massless scalar field on asymptotically flat spacetimes \cite{dappiaggi:2006cg}, and for both scalar and Dirac fields on a large class of cosmological backgrounds \cite{dappiaggi:2011ya,dappiaggi:2009th,dappiaggi:2009fk}.

In this paper we will focus instead on the electromagnetic field, which is usually described by a one-form, the vector potential.
Our chief goal will be to show that it is possible to explicitly construct Hadamard states employing the bulk-to-boundary procedure just outlined if the underlying background is asymptotically flat and globally hyperbolic.
It is worth mentioning that, up to now, there are no explicit examples of Hadamard states for the vector potential on a non-trivial background and even their existence has only recently been proven in \cite{fewster:2003or} although under some additional restrictive hypotheses on the topology of the underlying background.
The reason for such deficiency cannot only be traced to the intrinsic difficulty of constructing Hadamard states, but also to the peculiarity of the vector potential.
Indeed, even in Minkowski spacetime and contrary to what happens with the scalar or with the Dirac field, there are problems arising with the standard quantization procedure of the vector potential \cite{strocchi:1967fk,strocchi:1970kx} which are at least partially overcome by employing the so-called Gupta-Bleuer formalism \cite{bleuler:1950fk,gupta:1950fk}.
We will not enter into a detailed analysis of this method, but we stress that, since it relies heavily on Fourier transforms, it is difficult to apply it when the underlying background is curved.
In this paper we will show that using a complete gauge fixing together with the bulk-to-boundary procedure provides an alternative method which is applicable also to a wider class of curved backgrounds.

The paper will be organized as follows: In section \ref{sec:preliminaries}, we will introduce our main notations and conventions and we will recollect the definition of an asymptotically flat spacetime with future timelike infinity.
In particular, we will summarize the main geometric properties of the conformal boundary emphasizing mostly the role of the BMS group and the properties of intrinsicness and universality of null infinity in section \ref{sub:geometry}.
Section \ref{sec:vecpot} focuses on the vector potential on curved backgrounds and section \ref{sub:dynamics} is fully devoted to the analysis of its classical dynamics following mostly the earlier works of \cite{dappiaggi:2011rt,dimock:1992ea,furlani:1995kx,pfenning:2009ti}.
We present our first novel result in section \ref{sub:quant}.
Namely, we will show that the vector potential can be described as a locally covariant conformal quantum field theory,\footnote{The quantization of the field strength tensor in the framework of general local covariance has been discussed in \cite{dappiaggi:2011yq}.} thus extending the result of \cite{dappiaggi:2011rt}.
Furthermore we will define the field algebra of this physical system.
We stress that a full-fledged quantization scheme for the vector potential has been already discussed in \cite{hollands:2008zr} where the BRST approach has been translated in the algebraic language.
We shall not work with this framework in this paper and we will stick to a more traditional approach.
Section \ref{sec:hadamard} is fully focused on the construction of quasi-free Hadamard states for the vector potential on a large class of asymptotically flat spacetimes.
In particular, we will introduce the bulk-to-boundary projection technique in section \ref{sub:bulk2boundary}.
We will prove that it can be applied to the vector potential although the gauge freedom has to be used in order to match the bulk algebra to a subalgebra of the boundary counterpart.
Hence in section \ref{sub:state} we construct explicitly a distinguished state for the boundary algebra proving that it induces a bulk state which is both Hadamard and invariant under bulk isometries.
We continue by showing that the boundary state can be slightly modified in order to define a one-parameter family of states on null infinity which fulfil an exact KMS condition.
The bulk counterpart turns out to enjoy the Hadamard property too.
In section \ref{sec:conclusions} we draw our conclusions.

\section{Preliminaries}\label{sec:preliminaries}

\subsection{Notations and Conventions}\label{sub:notations}

In the following paragraphs we shall introduce the main technical tools that we use throughout the paper.

We shall call \emph{spacetime} $(M, g)$ a four dimensional, Hausdorff, second countable, arcwise connected, smooth manifold $M$ endowed with a smooth Lorentzian metric $g$ of signature $(-,+,+,+)$.
On top of the geometric structure we shall consider $\Omega^p(M, \mathbb{K})$ and $\Omega^p_0(M, \mathbb{K})$, respectively the set of smooth and of smooth and compactly supported $p$-forms on $M$ with values in the field $\mathbb{K} = \mathbb{R}$ or $\mathbb{C}$.
In the case $\mathbb{K} = \mathbb{R}$ we omit its explicit mention since no confusion can arise.

If we require $M$ to be orientable, on top of these spaces one can define two natural operators: the \emph{external derivative} $d$ and the \emph{Hodge dual} $*$. Notice that, whereas $d$ is completely independent from $g$, $*$ is the unique isomorphism $\Omega^p_0(M, \mathbb{K}) \to \Omega^{4-p}_0(M, \mathbb{K})$ built upon the metric such that $\omega \wedge *\, \eta = g^{-1}(\omega, \eta)\, \mu_g$ for all $\omega, \eta \in \Omega^p(M, \mathbb{K})$. Here $g^{-1}(\omega,\eta)$ and $\mu_g$ are natural pairing between $p$-forms and the volume form induced by the metric tensor $g$ respectively.
Furthermore we can introduce a third operator, known as the \emph{codifferential} $\delta \doteq *\, d\, * : \Omega^{p}(M, \mathbb{K}) \to \Omega^{p-1}(M, \mathbb{K})$, which is the formal adjoint of $d$ with respect to the metric $L^2$-pairing $\int (\cdot \wedge * \, \cdot)$.
We can then combine $d$ and $\delta$ to define the (formally self-adjoint) \emph{Laplace-de Rham operator} $\square \doteq d \circ \delta + \delta \circ d$.

In the main body of the paper we will often be interested in (compactly supported) smooth forms which are either closed or coclosed.
Therefore we introduce the following non-standard notation:
\begin{align*}
  \Omega^p_{(0),\delta}(M, \mathbb{K}) &\doteq \big\{ \omega \in \Omega^p_{(0)}(M, \mathbb{K}) \mid \delta \omega = 0 \big\},\\
  \Omega^p_{(0),d}(M, \mathbb{K}) &\doteq \big\{ \omega \in \Omega^p_{(0)}(M, \mathbb{K}) \mid d \omega = 0 \big\}.
\end{align*}

Another ingredient which we shall need in the forthcoming discussion is
$H^p(M)$, the \emph{$p$-th de Rham cohomology group} of $M$ -- see \textit{e.g.} \cite{bott:1982hz} for a definition and a recollection of the main properties.
It is noteworthy that, since these groups are built only out of the external derivative, they are homotopy invariants and, in particular, completely independent from a metric structure.

In the next section, we want to describe the dynamical behaviour of an electromagnetic field without sources on a spacetime $(M, g)$.
Therefore we have to make sure that it is possible to define an initial value problem for Maxwell's equations and, thus, we require $(M, g)$ to be \emph{globally hyperbolic}.
This implies both that $M$ is orientable and time orientable and that there exists $\Sigma$, a closed achronal subset of $M$ whose domain of dependence coincides with the whole manifold.
As proven in \cite{bernal:2003jf,bernal:2006oq}, this suffices to guarantee that $\Sigma$ can be chosen as a three dimensional smooth embedded hypersurface, called \emph{Cauchy surface}, and that $M$ is, moreover, isometric to the smooth product manifold $\mathbb{R}\times\Sigma$.

Since the Laplace-de Rham operator $\Box$ is normally hyperbolic, we can find unique retarded and advanced fundamental solutions $G^\pm : \Omega^p_0(M, \mathbb{K}) \to \Omega^p(M, \mathbb{K})$ such that $\Box \circ G^\pm = G^\pm \circ \Box = \mathrm{id}$ and $\supp G^\pm(f) \subseteq J^\mp(\supp f)$ whose domain of definition may be extended to distributions \cite[Prop. 3.4.2, Cor. 3.4.3]{bar:2007fi}.
Further, we define the \emph{causal propagator} $G \doteq G^+ - G^-$ as the difference between the retarded and advanced fundamental solution.

\subsection{Conformal Transformations}\label{sub:conformal}

Let us consider two arbitrary spacetimes $(M, g)$ and $(\widetilde M, \tilde g)$ of the same dimension.
Following partly \cite[Appendix C \& D]{wald:1984gp}, we say that a map $\psi : M \to \widetilde M$ is a \emph{conformal embedding} if it is a smooth diffeomorphism from $M$ to $\psi(M)$ and if there exists a strictly positive function $\Xi \in C^\infty(M)$ such that $\psi^* \tilde g = \Xi^2 g$.
If $\psi(M)$ coincides with $\widetilde M$, then it is referred to as \emph{conformal isometry}.

In order to understand the effect of such $\psi$ on classical fields, we recall that these are best understood as sections of a suitable vector bundle.
Let us therefore choose a vector bundle $\pi : E \to \widetilde{M}$.
A conformal isometry $\psi : M \to \widetilde{M}$ yields automatically on $M$ the pull-back bundle $\psi^* E$, so that, if we refer to $\Gamma(E)$ as the space of smooth sections, an element $\tilde s \in \Gamma(E)$ induces $s = \psi^* \tilde s \in \Gamma(\psi^* E)$.
Let us now introduce the \emph{conformally weighted pull-back} and \emph{push-forward} of weight $w$ by defining
\begin{subequations}\label{confpullback}
\begin{align}
  \psi^*_{(w)} \doteq \Xi^{-w} \psi^*\\
  \psi_{(w)\;*} \doteq \psi_* \Xi^w
\end{align} .
\end{subequations}
A section $s' \in \Gamma(\psi^* E)$ which is given by $s' = \psi^*_w \tilde s$ is said to be of \emph{conformal weight} $w$.

In the forthcoming discussion of the vector potential, an important role will be played by the behaviour of the operators $d$, $\delta$ and $\Box$ under conformal isometries:

\begin{proposition}\label{conformaltransf}
  Let $(M, g)$, $(\widetilde M, \tilde g)$ be two spacetimes such that $\psi : M \to \widetilde M$ is a conformal isometry with $\psi^* \tilde g = \Xi^2 g$.
  Further, let the codifferential on $(M,g)$ and $(\widetilde M,\tilde g)$ be called $\delta$ and $\tilde\delta$ respectively and let $\tilde\phi\in\Omega^1(\widetilde M)$.
  Then, it holds that
  \begin{subequations}
    \begin{align}
      \psi^* \tilde \delta \tilde \phi
      &= \Xi^{-2} \big( \delta \phi - 2\, g^{-1}(\Upsilon, \phi) \big), \label{delta}\\
      \psi^* \tilde \delta d\, \tilde \phi
      &= \Xi^{-2} \delta d \phi, \label{deltad}\\
      \psi^* \tilde \Box \tilde \phi
      &= \Xi^{-2} \big( \Box \phi - 2\, d\, g^{-1}(\Upsilon, \phi) + 2\, \Upsilon \wedge (2\, g^{-1}(\Upsilon, \phi) - \delta \phi) \big), \label{confbox}
    \end{align}
  \end{subequations}
  where $\phi \doteq \psi^* \tilde\phi$, $\Upsilon \doteq \Xi^{-1} d\Xi$ and $g^{-1}(\cdot\,,\cdot)$ is the metric pairing of $1$-forms.
\end{proposition}
\begin{proof}
  Since the exterior derivative depends only on the differentiable structure of the manifold, it holds that $\psi^* d \omega = d \psi^* \omega$ for any $\omega \in \Omega^p(\widetilde M)$.
  On the contrary, the Hodge operator is built out of the metric and, thus, on account of its definition,
  \begin{equation}\label{confhodge}
    \psi^*\, \tilde{*}\, \omega = *\; \Xi^{4-2p} \psi^* \omega
  \end{equation}
  for all $\omega \in \Omega^{p}(\widetilde M)$.
  If we put these ingredients together with the properties of the pull-back of forms and with $\delta = *\, d\, *$, \eqref{delta} follows immediately.

  To prove the third equation, we stress that the properties of $d$ and $*$ under a conformal isometry entail that $\psi^* (d \, \tilde{*} \, d \tilde \phi) = d * d \phi$.
  Since $d \, \tilde{*} \, d \tilde \phi \in \Omega^{3}(M)$, \eqref{deltad} follows immediately from \eqref{confhodge} with $p = 3$.

  In order to show equality in \eqref{confbox}, recall that $\tilde \square = \tilde \delta d + d \tilde{\delta}$.
  Hence we are left to evaluate
  $\psi^*(d \tilde{\delta} \tilde \phi)$.
  On account of the properties of the exterior derivative and of the codifferential, the following chain of identities holds true:
  \begin{align*}
    \psi^* ( d \tilde{\delta} \tilde \phi )
    & = d \psi^* ( \tilde{\delta} \tilde \phi )
      = d \big( \Xi^{-2} ( \delta \phi - 2\, g^{-1}(\Upsilon, \phi) ) \big) \\
    & = \Xi^{-2} d \delta \phi - 2\, \Xi^{-2} \Upsilon \wedge \delta \phi - 2\, d\, g^{-1}(\Upsilon, \phi),
  \end{align*}
  where, in the second equality, we employed \eqref{delta}.
  This result combined with \eqref{deltad} yields the sought \eqref{confbox}.
\end{proof}

\noindent The following lemma holds:
\begin{lemma}\label{new-lemma}
  Under the same hypotheses in proposition \ref{conformaltransf}, it holds that for all $\tilde\phi\in\Omega^1(\widetilde M)$
  \begin{equation}\label{delta2}
    \psi^*_{(-4)} \tilde \delta \tilde \phi = \delta \psi^*_{(-2)}\tilde\phi.
  \end{equation}
\end{lemma}
\begin{proof}
  Writing $\phi = \Xi^{-2} \psi^*_{(-2)} \tilde\phi$ on the right hand side of \eqref{delta}, we obtain $\Xi^{-2} ( \delta \Xi^{-2} \psi^*_{(-2)} \tilde\phi - 2\, g^{-1}(\Upsilon, \Xi^{-2} \psi^*_{(-2)} \tilde\phi ) )$.
  An application of the identity $\delta (f \omega) = f \delta \omega - g^{-1}(d f, \omega)$ for all $f \in C^\infty(M)$ and $\omega \in \Omega^1(M)$ yields $\psi^* \tilde \delta \tilde \varphi = \Xi^{-4} \delta \psi^*_{(-2)} \tilde\varphi$ which can be rewritten into the form of \eqref{delta2} using definition \eqref{confpullback}.
\end{proof}

On account of these results, we say that the codifferential $\delta$ and the operator $\delta d$ acting on $1$-forms are \emph{conformally invariant} under the action of the conformally weighted pull-back (or push-forward) of weight $0$ and $-2$ respectively.

\subsection{Asymptotically Flat Spacetimes}\label{sub:asymp_flat}

In this paper we will be mostly interested in a subclass of globally hyperbolic spacetimes $(M, g)$ which are solutions of Einstein vacuum equations and whose behaviour at infinity along null geodesics mimics that of Minkowski spacetime.
Therefore let us consider an \emph{asymptotically flat spacetime with future timelike infinity $i^+$}, \emph{i.e.}, a time-oriented globally hyperbolic spacetime $(M, g)$, called \emph{physical spacetime}, such that there exists a second globally hyperbolic spacetime $(\widetilde M, \tilde g)$, called \emph{unphysical spacetime}, with a preferred point $i^+$, a diffeomorphism $\psi : M \to \psi(M) \subset \widetilde M$ and a function $\Xi : \psi(M) \to (0, \infty)$ such that $\tilde g = \Xi^2 \, \psi_* g$.
Moreover the following requirements are satisfied:\footnote{$\tilde \nabla$ is the Levi-Civita connection built out of $\tilde g$.}
\begin{itemize}
  \item[a)]
    If we call $J^-(i^+)$ the causal past of $i^+$, this is a closed set such that $\psi(M) = J^-(i^+) \setminus \partial J^-(i^+)$ and we have $\partial M = \partial J^-(i^+) = \mathscr{I}^+ \cup \{i^+\}$, where $\mathscr{I}^+$ is \emph{future null infinity}.
  \item[b)]
    $\Xi$ can be extended to a smooth function on the whole of $\widetilde M$ and it vanishes on $\mathscr{I}^+ \cup \{i^+\}$.
    Furthermore $d \Xi \neq 0$ on $\mathscr{I}^+$ while $d \Xi = 0$ on $i^+$ and $\tilde \nabla_\mu \tilde \nabla_\nu \, \Xi = -2 \, \tilde g_{\mu\nu}$ at $i^+$.
  \item[c)]
    If $n^\mu \doteq \tilde g^{\mu\nu} \partial_\nu \Xi$, there exists a smooth and positive function $\xi$ supported at least in a neighbourhood of $\mathscr{I}^+$ such that $\tilde \nabla_\mu (\xi^4 n^\mu) = 0$ on $\mathscr{I}^+$ and the integral curves of $\xi^{-1} n$ are complete on future null infinity.
\end{itemize}

It is worth remarking that this definition is different from the standard one (\textit{cf.}, \textit{e.g.} \cite[Chap. 11]{wald:1984gp}), where asymptotic flatness is defined with respect to $i_0$, spatial infinity, as distinguished point.
The underlying reason for adopting such a choice is related to our need to describe solutions of a second order hyperbolic partial differential operator on such backgrounds.
If smooth and compactly supported initial data are assigned, the associated solutions will be supported in the causal future and past of the initial data and thus it is important that either $i^+$ or $i^-$ are included in order to avoid any loss of information.
The above definition appeared already in \cite{moretti:2006ap,moretti:2008ay}, where the equivalence with the definition proposed by Friedrich (see \cite{friedrich:1986uq} and references therein) is also pointed out.
Notice that, in comparison with these last cited papers, we dropped the requirement of strong causality for $M$ and $\widetilde M$ since it is a property enjoyed by all globally hyperbolic spacetimes.
Moreover, although one could equivalently define an asymptotically flat spacetime with past timelike infinity $i^-$, we will not mention this possibility again since all our results can be easily translated to that case.

\subsection{Geometric Properties of Null Infinity}\label{sub:geometry}

For our purposes the most notable property of an asymptotically flat spacetime is the existence of future null infinity which is a three dimensional submanifold of $\widetilde M$ generated by the null geodesics emanating from $i^+$, {\it i.e.}, the integral curve of $n$.
For this reason $\mathscr{I}^+$ is diffeomorphic to $\mathbb{R} \times \mathbb{S}^2$ although the possible metric structures are affected by the existence of a gauge freedom which corresponds to the rescaling of $\Xi$ to $\xi \Xi$, where $\xi$ is a smooth function which is strictly positive in $\psi(M)$ and a neighbourhood of $\mathscr{I}^+$.
Furthermore, if we introduce for any fixed asymptotically flat spacetime $(M, g)$ the set $C$ composed of equivalence classes of triples $(\mathscr{I}^+, h, n)$, where $h \doteq \tilde{g} \restriction_{\mathscr{I}^+}$ and $(\mathscr{I}^+, h, n) \sim (\mathscr{I}^+, \xi^2 h, \xi^{-1} n)$ for any choice of $\xi$ satisfying c), there is no physical mean to select a preferred element in $C$.
This is called the \emph{intrinsicness} of $\mathscr{I}^+$ and it emphasises the relevance of the boundary structure together with the property of \emph{universality}.
Namely, if we select any pair of asymptotically flat spacetimes, $(M_1, g_1)$ and $(M_2, g_2)$, together with the corresponding triples, say $(\mathscr{I}^+_1, h_1, n_1)$ and $(\mathscr{I}^+_2, h_2, n_2)$, there always exists a diffeomorphism $\gamma : \mathscr{I}^+_1 \to \mathscr{I}^+_2$ such that $h_1 = \gamma^* h_2$ and $n_2 = \gamma_* n_1$.
Although we leave a reader interested in the proof of this last statement to \cite[Chap. 11]{wald:1984gp}, it is noteworthy that it relies on a very important additional property of an asymptotically flat spacetime:
In each equivalence class there exists a choice of conformal gauge $\xi_B$ yielding a coordinate system $(u, \Xi, \theta, \varphi)$ in a neighbourhood of $\mathscr{I}^+$, called a \emph{Bondi frame}. The (rescaled) unphysical metric tensor becomes
\begin{equation*}
  \tilde g \restriction_{\mathscr{I}^+} = -2\, du\, d\Xi + d\theta^2 + \sin^2 \! \theta\, d\varphi^2.
\end{equation*}
That is, $\Xi$ is promoted from conformal factor to a coordinate, thus indicating that future null infinity is the locus $\Xi = 0$, while $u$ is the affine parameter of the null geodesics generating $\mathscr{I}^+$.
Thus at each point on $\mathscr{I}^+$ the vector field $n$ coincides with $\partial / \partial u$.

A very important role is played by the subgroup of diffeomorphisms of $\mathscr{I}^+$ which maps any triple $(\mathscr{I}^+, h, n)$ into a gauge equivalent one.
This is the so-called \emph{Bondi-Metzner-Sachs} (BMS) group of transformations whose structure is that of semidirect product, \textit{i.e.}, $BMS = SO_0(3,1) \rtimes C^\infty(\mathbb{S}^2)$.
The action of each element $\gamma \in SO_0(3,1) \rtimes C^\infty(\mathbb{S}^2)$ is best represented in a Bondi frame, that is, going over to the complex stereographic coordinate $z = e^{i \varphi} \cot \theta / 2$, we have the mapping
\begin{equation}\label{BMS}
  \begin{split}
    u & \mapsto u' = K_\gamma(z, \bar{z})(u + \alpha_\gamma(z, \bar{z})),\\
    z & \mapsto z' = \frac{az+b}{cz+d} \quad \text{and c.c.,}\\
  \end{split}
\end{equation}
where
\begin{equation*}
  K_\gamma(z, \bar{z}) = \frac{1 + |z|^2}{|az+b|^2 + |cz+d|^2},
  \quad
  \alpha_\gamma \in C^\infty(\mathbb{S}^2),
\end{equation*}
and $a,b,c,d \in \mathbb{C}$ with $ad - bc = 1$, \textit{i.e.}, the $SO_0(3,1) \cong PSL(2, \mathbb{C}) \cong \mathrm{Aut}(\mathbb{C} \cup \{\infty\})$ part is represented via M\"obius transformations.
The elements of the form $(I,\alpha)$ where $I$ is the identity element of $SO_0(3,1)$ are called {\em supertranslations} and they form a proper Abelian normal subgroup of the BMS group, homeomorphic to $C^\infty(\mathbb{S}^2)$, seen as an Abelian group under addition.

On account of the representation of the BMS group in stereographic coordinates, we can now consider a family of representations of the BMS group on sections of a vector bundle $E \to \mathscr{I}^+$.
Namely, for each $w \in \mathbb{R}$ we have the representation $\Pi^w : \Gamma(E) \rightarrow \Gamma(E)$ defined by
\begin{equation}\label{BMSrep}
  (\Pi^w_\gamma s)(u', z', \bar{z}')
  \doteq K_\gamma(z, \bar{z})^w s(u + \alpha_\gamma(z, \bar{z}), z, \bar{z}),
\end{equation}
where $K_\gamma, \alpha_\gamma$ as above, for each $\gamma \in BMS$ and all $s \in \Gamma(E)$.

From the geometrical point of view the BMS group can be interpreted as the group of asymptotic symmetries for all asymptotically flat spacetimes as shown in \cite[Chap. 11]{wald:1984gp}.
This entails that each one-parameter subgroup of diffeomorphisms of $\mathscr{I}^+$ generated by a smooth vector field $X'$ is a subgroup of the BMS group if and only if $X'$ can be extended smoothly but not necessarily in a unique way to a vector field $X$ of the bulk spacetime $(M, g)$ in such a way that $\Xi^2 \mathcal{L}_X g$ admits a smooth and vanishing extension to $\mathscr{I}^+$.
In particular, as proven by Geroch in \cite{geroch:1977fv}, if $X$ is a Killing vector field in the physical spacetime, this guarantees the existence of a unique extension to future null infinity.
Thus it is commonly said that the BMS group encodes the information of all possible bulk symmetries.

\section{The Vector Potential on Curved Spacetimes}\label{sec:vecpot}

In this section we will introduce the main object of our studies, the vector potential.
We will discuss its classical kinematical and dynamical configurations and outline its quantization in the algebraic approach.
We will emphasise in particular how this physical system can be described in terms of a locally covariant conformal quantum field theory.
Part of the material we shall present has already appeared in the literature and, in particular, we will benefit of earlier analyses \cite{dappiaggi:2011rt,dappiaggi:2011yq,dimock:1992ea,fewster:2003or,pfenning:2009ti,siemssen:2011fk}.

\subsection{Classical Dynamics}\label{sub:dynamics}

Classical electromagnetism on a globally hyperbolic spacetime $(M, g)$ is assumed to be described by the natural extension to curved backgrounds of the manifestly covariant version of Maxwell's equations on Minkowski spacetime.
In other words, one should consider the field strength tensor $F \in \Omega^2(M)$, which satisfies
\begin{equation*}
  dF = 0,
  \quad
  \delta F=-j,
\end{equation*}
where the source $j$ is a coclosed smooth $1$-form.
In the forthcoming discussion we shall always assume that $j = 0$ and hence we only consider the Maxwell equations without sources.

The dynamics provided by this system is best analysed in terms of the vector potential, that is, $A \in \Omega^1(M)$ such that $F = dA$.
If $H^2(M)$, the second de Rham cohomology group of $M$, is trivial, the global existence of this $1$-form is per definition guaranteed.
Otherwise there exist field strength tensors which solve Maxwell's equations and do not descend from a vector potential.
Since the goal of this paper is to focus on the construction of Hadamard states for $A$, we shall not dwell on this problem.
Nevertheless, the reader is warned of this obstruction and, if interested in a discussion of this issue, should refer to \cite{dappiaggi:2011yq}.
Hence, in terms of the vector potential, the dynamics is ruled by
\begin{equation}\label{Proca}
  \delta d A = 0.
\end{equation}
In an arbitrary coordinate system of $(M, g)$ this reads
\begin{equation}\label{Procalocal}
  - \nabla_\nu \nabla^\nu A_\mu + \nabla_\mu \nabla^\nu A_\nu + R_\mu^\nu A_\nu = 0,
\end{equation}
where $R_{\mu\nu}$ is the Ricci tensor built out of $g$.
Inspecting the principal symbol of \eqref{Procalocal}, one can realise that the differential operator is not of hyperbolic type and thus one cannot straightforwardly set up an initial value problem.
This difficulty can be solved if we recall that two vector potentials can yield the same field strength if they are \emph{gauge equivalent}.
The ensuing equivalence classes are
\begin{equation*}
  [A] = \big\{ A' \in \Omega^1(M) \mid A' \sim A, \text{ iff } \exists \Lambda \in \Omega^1_d(M) \text{ with } A - A' = \Lambda \big\}.
\end{equation*}
That is, two vector potentials are gauge equivalent and thus in the same equivalence class if they differ by a closed $1$-form.
From a physical point view this choice might appear as rather unorthodox since it is customary to identify two vector potentials differing by the differential of a smooth function on $M$.
We stress that one should keep in mind that the physical content of Maxwell's equations is encoded in the Faraday tensor $F$ and in its equations of motion.
This entails that one should regard as equivalent two vector potentials yielding the same field strength tensor.
Our choice for the equivalence classes $[A]$ reflects this fact from a mathematical point of view.
Whenever $H^1(M)=\{0\}$ no difference arises with the procedure usually employed in Minkowski spacetime.
Notice also that, in this analysis, no role is played by the fact that the first Maxwell equation, $dF=0$, is equivalent to the existence of a smooth $1$-form $A$, such that $F=dA$, only if $H^2(M)=\{0\}$.

As proven for example in \cite{dappiaggi:2011rt}, the following lemma holds true:

\begin{lemma}\label{gaugeequiv}
  Each $1$-form $A'$ satisfying $\delta d A' = 0$ is gauge equivalent to a vector potential $A$ such that
  \begin{equation}\label{dynamics}
    \square A = 0,
    \quad
    \delta A = 0.
  \end{equation}
\end{lemma}

In an arbitrary coordinate system \eqref{dynamics} reads $- \nabla_\nu \nabla^\nu A_\mu + R_\mu^\nu A_\nu = 0$ together with $\nabla^\mu A_\mu = 0$.
In other words, the dynamics of $A$ is ruled by a second order hyperbolic partial differential operator with principal symbol of metric type and the vector potential is further constrained to satisfy the \emph{Lorenz gauge}.
Notice that a residual gauge freedom is present, since it is still possible to add to any Lorenz solution $A$ a smooth $1$-form $\Lambda$ which is at the same time closed and coclosed.
Thanks to this freedom we can consider solutions of \eqref{dynamics} with smooth and compactly supported initial data.
These are called \emph{Lorenz solution} and they can be written as $A = G f$ where $f \in \Omega^1_{0,\delta}(M)$ and $G$ is the causal propagator associated to the Laplace-de Rham operator $\Box$.
Notice that the coclosedness of $f$ suffices to guarantee that $A$ fulfils the Lorenz gauge because, as already remarked in \cite{dimock:1992ea,fewster:2003or,pfenning:2009ti}, $G \circ \delta = \delta \circ G$, where, on the left hand side, $G$ is meant as the causal propagator for $\Box$ acting on $p$-forms whereas, on the right hand side, it is the one acting on $(p-1)$-forms.

Notice that the choice of coclosed $1$-forms as initial data does not account for the whole residual gauge freedom.
The following proposition clarifies under which conditions on the initial data this can happen.

\begin{proposition}\label{equivalence}
  Let $f, f' \in \Omega^1_{0,\delta}(M)$.
  Then $G f \sim G f'$ if and only if there exists $\lambda \in \Omega^2_{0,d}(M)$ such that $f - f' = \delta \lambda$.
\end{proposition}
\begin{proof}
  Suppose $f - f' = \delta \lambda$ for $f, f'\in\Omega^1_{0,\delta}(M)$.
  Then $\Lambda \doteq G f - G f' = G(\delta \lambda)$ is such that
  \begin{equation*}
    d \Lambda = d G(\delta \lambda) = G(d \delta \lambda) = G(\square \lambda) = 0,
  \end{equation*}
  where we employed the definition of the Laplace-de Rham operator and closedness of $\lambda$ in the second equality.

  Conversely, suppose that two initial data $f, f' \in \Omega^1_{0,\delta}(M)$ are such that the associated Lorenz solutions are gauge equivalent.
  Therefore $G(f - f')$ is closed, \textit{i.e.}, $G(d(f-f')) = 0$, which in turn guarantees the existence of $\lambda \in \Omega^2_0(M)$ such that $d (f - f') = \Box \lambda$.
  If one applies to both sides the external derivative, the ensuing identity, $d \Box \lambda = \Box d \lambda = 0$, entails $d \lambda = 0$ as $\lambda$ is compactly supported.
  Furthermore,
  \begin{equation*}
    \Box \delta \lambda = \delta \Box \lambda = \delta d (f - f') = \Box(f - f'),
  \end{equation*}
  where, in the last equality, we exploited the coclosedness of both $f$ and $f'$.
  Since all forms involved are compactly supported, the above chain of identities yields that $\delta \lambda = f - f'$.
\end{proof}

We denote the collection of these equivalence classes of test forms associated to Lorenz solutions
\begin{equation*}
  \mathcal{S}(M) \doteq \frac{\Omega^1_{0,\delta}(M)}{\delta\Omega^2_{0,d}(M)}.
\end{equation*}
Notice then, that the bottom line of proposition \ref{equivalence} is that $\mathcal{S}(M)$ is in one-to-one correspondence with the gauge equivalence classes of solutions of \eqref{Proca}.

Although proposition \ref{equivalence} applies to all globally hyperbolic spacetimes, it is possible to derive a more stringent result if some additional assumptions on the topology of $(M, g)$ are made.
As we shall see, this will play an important role in the discussion of the quantization of the theory.

\begin{corollary}\label{cohomology}
  Let $(M, g)$ be a globally hyperbolic spacetime such that either $H^1(M)$ or $H^2(M)$ is trivial.
  Then two Lorenz solutions $G f$ and $G f'$ are equivalent if and only if there exist $\alpha \in \Omega^1_0(M)$ and $\chi \in C^\infty(M)$ such that
  \begin{equation*}
    f - f' = \delta d \alpha,
    \quad
    G(f - f') = d \chi.
  \end{equation*}
\end{corollary}
\begin{proof}
  Let us start supposing $H^1(M)=\{0\}$.
  Then two gauge equivalent Lorenz solutions $G f $ and $G f'$ differ by a closed one form which suffices to guarantee the existence of $\chi \in C^\infty(M)$ such that $G (f - f') = d \chi$.
  Since $\delta G(f - f') = 0$, it holds that $\Box \chi = 0$ which implies that there exists $\eta \in C^\infty_0(M)$ such that $G \eta = \chi$.
  Hence $G (f - f' - d \eta) = 0$ and one can find $\alpha \in \Omega^1_0(M)$ such that $f - f' = d \eta + \Box \alpha$.
  Since both $f$ and $f'$ are coclosed, also $\Box (\eta + \delta \alpha)$ vanishes, which is tantamount to $\eta = - \delta \alpha$ and, thus, to $f - f' = \delta d \alpha$.

  Conversely, suppose now that $H^2(M)$ is trivial.
  On account of proposition \ref{equivalence} we know that two Lorenz solutions $G f$ and $G f'$ are gauge equivalent if $f - f' = \delta \lambda$ with $\lambda \in \Omega^2_{0,d}(M)$.
  Hence, by Poincar\'e duality, there exists $\alpha \in \Omega^1_0(M)$ such that $\lambda = d \alpha$ and we obtain $f - f' = \delta d \alpha$.
  If we apply to both sides the causal propagator, we end up with $G(f - f') = G(\delta d \alpha) = - G(d \delta \alpha) = - d G(\delta \alpha)$ which is the sought result if we set $\chi = - G(\delta \alpha)$.
\end{proof}

\noindent According to corollary \ref{cohomology}, we obtain
\begin{equation*}
  \mathcal{S}(M) \cong \frac{\Omega^1_{0,\delta}(M)}{\delta d \Omega^1_0(M)}
\end{equation*}
whenever either $H^1(M)$ or $H^2(M)$ is trivial.
Notice that, on account of \eqref{delta}, the elements of $\mathcal{S}(M)$ have conformal weight $-2$ and \eqref{confpullback} entails that the space of compactly supported coclosed smooth $1$-forms of conformal weight $-2$ is conformally covariant.

To conclude our analysis of the classical theory of the vector potential, we need a last result which gives $\mathcal{S}(M)$ the structure of a phase space and, at the same time, creates a bridge towards the quantization of the theory.

\begin{proposition}\label{symplform}
  The set $\mathcal{S}(M)$ is a weakly non-degenerate symplectic space if either $H^1(M)$ or $H^2(M)$ is trivial and if endowed with the antisymmetric bilinear form $\sigma : \mathcal{S}(M) \times \mathcal{S}(M) \to \mathbb{R}$,
  \begin{equation}\label{symplectic}
    \sigma([f],[h]) \doteq G(f \otimes h) = \int\limits_M G f \wedge *\, h.
  \end{equation}
  Here, $[f], [h] \in \mathcal{S}(M)$ with $f, h \in \Omega^1_{0,\delta}(M)$ being two arbitrary representatives of these equivalence classes.
\end{proposition}
\begin{proof}
  As a first step we prove that \eqref{symplectic} is well-defined, \textit{i.e.}, it is independent from the chosen representative.
  Let $h, h' \in [h]$.
  Then the following chain of identities holds true:
  \begin{equation*}
    \sigma([f],[0])
    = \int\limits_M G f \wedge * \, (h - h')
    = \int\limits_M G f \wedge * \, \delta d \alpha
    = \int\limits_M \delta d G f\wedge * \, \alpha
    = 0,
  \end{equation*}
  where in the second equality we exploited corollary \ref{cohomology}, while, in the last, we used the properties of $G$ and the coclosedness of $f$.
  Since $\sigma$ is bilinear and antisymmetric, this suffices to prove the independence of both entries from the chosen representative.

  To prove weak non-degenerateness, consider $h = \delta \beta$ with $\beta \in \Omega^2_0(M)$ arbitrary.
  Then
  \begin{equation*}
    \sigma([f],[h])
    = \int\limits_M G f \wedge * \, h
    = \int\limits_M G f \wedge * \, \delta \beta
    = \int\limits_M d G f \wedge * \, \beta
    = 0
  \end{equation*}
  implies that $d G f = 0$, \textit{i.e.}, $G f$ is a pure gauge solution, since the metric pairing is non-degenerate.
  Thus $f \in [0]$ by corollary \ref{cohomology}.
\end{proof}

It is worthwhile to remark that this last proposition extends the results of \cite{dappiaggi:2011rt}, where the same statement was proven under the hypothesis that $H^1(M)$ is trivial.
If the underlying manifold does not fulfil the assumptions on the topology, it turns out that it is not clear whether \eqref{symplectic} is independent from the choice of the representative of the equivalence classes involved.
Notice that it is easy to construct globally hyperbolic spacetimes which do not fulfil the hypotheses of proposition \ref{symplform}.
For example, any ultrastatic manifold $(M, g)$ isometric to $\mathbb{R} \times \mathbb{S}^1 \times \mathbb{S}^2$ with line element $ds^2 = -dt^2 + h$, where $t$ is a global time-coordinate and $h$ is a smooth, Riemannian and time-independent metric on the Cauchy surface $\mathbb{S}^1 \times \mathbb{S}^2$, has non-trivial first and second de Rham cohomology group.

It is important to mention that the origin of this problem could be traced to the set of equivalence classes we introduced.
In particular we identify two vector potentials differing by a closed $1$-form.
If we would have followed a more conservative point of view replacing closed forms with exact ones, then \eqref{symplectic} would turn out to be a well-defined weakly non-degenerate symplectic form, but we lose the connection to the field strength tensor.

Let us conclude this subsection by investigating some of the properties of the vector potential under conformal transformations.
As a first result we have the following corollary which is a direct consequence of \eqref{deltad}.

\begin{corollary}\label{confvecpot}
  Maxwell's equation $\eqref{Proca}$ is conformally invariant on $1$-forms of conformal weight $0$.
  That is, given two conformally isometric spacetimes $(M, g)$, $(\widetilde M, \tilde g)$ such that $\psi : M \rightarrow \widetilde M$ with $\psi^* \tilde g = \Xi^2 g$, then $A \in \Omega^1(M)$ is a solution of $\eqref{Proca}$ in $(M, g)$ if and only if $\tilde A \in \Omega^1(\widetilde M)$ solves $\eqref{Proca}$ in $(\widetilde M, \tilde g)$ and $A = \psi^* \tilde A$.
\end{corollary}

Notice that the properties of Lorenz solutions under conformal transformations on a special class of curved background were also emphasised in \cite{queva:2009pd}.

In this subsection we studied the equation $\eqref{Proca}$ by considering the equivalence classes of Lorenz solutions, \textit{i.e.}, solutions of the constrained hyperbolic system $\eqref{dynamics}$.
The Lorenz gauge condition, however, is not conformally invariant on $1$-forms under pull-backs resp. push-forwards as can be readily seen from \eqref{delta}.
Nevertheless we can obtain the following important result.

\begin{proposition}\label{confvecpot2}
  Let $(M, g)$, $(\widetilde M, \tilde g)$ be globally hyperbolic spacetimes such that $\psi : M \rightarrow \widetilde M$ is a conformal isometry with $\psi^* \tilde g = \Xi^2 g$.
  Further, let $\tilde A = \tilde G \tilde f$ with $\tilde f \in \Omega^1_{0,\delta}(\widetilde M)$ be a solution of $\eqref{dynamics}$ in $(\widetilde M, \tilde g)$ (\textit{i.e.} $\tilde G$ is the causal propagator associated to the Laplace-de Rham operator $\tilde \Box$ in $(\widetilde M, \tilde g)$).
  Then,
  \begin{equation*}
    A = G f = G \psi^*_{(-2)} \tilde f
    \quad \text{with} \quad
    f = \psi^*_{(-2)} \tilde f \in \Omega^1_{0,\delta}(M)
  \end{equation*}
  solves $\eqref{dynamics}$ in $(M, g)$ and $\psi^* \tilde A = A + d\lambda$ for some $\lambda \in C^\infty(M)$.
\end{proposition}
\begin{proof}
Notice that, on account of lemma \ref{new-lemma}, $f=\psi^*_{(-2)}\tilde f$ is a coclosed compactly supported $1$-form on $M$ such that $A$ is a solution of \eqref{Proca}.
We obtain via \eqref{deltad}
  \begin{equation*}
    \psi^*_{(-2)} \tilde f
    = \psi^*_{(-2)} \tilde \delta d \tilde A^\pm
    = \delta d \psi^* \tilde A^\pm
    = f,
  \end{equation*}
  where $\tilde A^\pm = \tilde G^\pm \tilde f$ and $f \in \Omega^1_{0, \delta}(M)$.
  Moreover, according to \eqref{delta} we have
  \begin{equation*}
    0 = \psi^*_{{(-2)}} \tilde \delta \tilde A^\pm = \delta \psi^* \tilde A^\pm - 2 g^{-1}(\Upsilon, \psi^* \tilde A^\pm).
  \end{equation*}
  Taking the exterior derivative, this yields $d \delta \psi^* \tilde A^\pm = 2 d g^{-1}(\Upsilon, \psi^* \tilde A^\pm)$ and thus $\Box \psi^* \tilde A^\pm = f + 2 d g^{-1}(\Upsilon, \psi^* \tilde A^\pm)$.
  Applying the advanced resp. the retarded Green's operator and subtracting both results, we finally obtain
  \begin{equation*}
    \psi^* \tilde A = A + 2 d \big( G^- g^{-1}(\Upsilon, \psi^* \tilde A^-) - G^+ g^{-1}(\Upsilon, \psi^* \tilde A^+) \big)
  \end{equation*}
  and, therefore, $\psi^* \tilde A = A + d\lambda$.
\end{proof}

That is, if we apply the unweighted push-forward of the conformal isometry to the vector potential (so that the equations of motions are conformally invariant), we have to apply the push-forward of weight $-2$ to the corresponding test form.

\begin{corollary}\label{confvecpot3}
  Let $(M, g)$, $(\widetilde M, \tilde g)$ and $\psi$ be as before.
  For any $\tilde f, \tilde h \in \Omega^1_{0,\delta}(\widetilde M)$ we have
  \begin{equation*}
    G(f \otimes h) = \tilde G(\tilde f \otimes \tilde h),
  \end{equation*}
  where $f = \psi^*_{(-2)} \tilde f$, $h = \psi^*_{(-2)} \tilde h$ and $G, \tilde G$ are the causal propagator on $(M, g)$, $(\widetilde M, \tilde g)$ respectively.
  That is, the symplectic form defined in proposition \ref{symplform} is invariant under conformal transformations.
\end{corollary}
\begin{proof}
  Denoting by $\tilde *$ the Hodge operator associated to $\tilde g$, we obtain
  \begin{equation*}
    \tilde G(\tilde f \otimes \tilde h)
    = \int_{M} \psi^* (\tilde G \tilde f \wedge \tilde *\, \tilde h)
    = \int_{M} \psi^* (\tilde G \tilde f ) \wedge *\, h
    = \int_{M} (G f + d\lambda) \wedge *\, h
  \end{equation*}
  since $\psi^* (\tilde G\tilde f) = G f + d \lambda$ for some $\lambda \in C^\infty(M)$ by proposition \ref{confvecpot2}.
  The second term, $\int_{M} d\lambda \wedge *\, h$, vanishes because $h$ is coclosed.
\end{proof}

\subsection{Quantization of the Vector Potential}\label{sub:quant}

The study of the classical dynamics of \eqref{Proca} via \eqref{dynamics} allows us as a by-product to tackle the problem of quantizing the underlying theory.
We shall approach this issue in the algebraic formalism via a two-step procedure.
First, we associate to the field under investigation a suitable $*$-algebra of observables which is compatible with the dynamics.
Then, in the second step, we endow the resulting algebra with a suitable (algebraic) quantum state.
We will leave this problem to the next section, while, in the present, we identify the so-called field algebra and we show that its structure allows us to describe the vector potential as a locally covariant conformal quantum field theory along the lines of \cite{brunetti:2003vo} and \cite{pinamonti:2009zn}.

In order to study the vector potential as a locally covariant conformal quantum field theory we need to introduce some additional mathematical tools.
We introduce the following categories:
\begin{itemize}
  \item
    $\mathfrak{CMan}$: the category whose objects are oriented and time-oriented globally hyperbolic spacetimes $(M, g)$.
    The arrows are conformal embeddings $\psi: M\hookrightarrow \widetilde M$ which preserve orientation and time-orientation and are such that $\psi(M)$ is a causally convex, hence globally hyperbolic, subset of $\widetilde M$.
    The composition of morphisms is that of smooth maps and the unit element is the identity map.
  \item
    $\mathfrak{CMan}'$: the subcategory of $\mathfrak{CMan}$ whose objects are oriented and time-oriented globally hyperbolic spacetimes $(M, g)$ with either $H^1(M)=\{0\}$ or $H^2(M)=\{0\}$.
    The arrows are the same as those of $\mathfrak{CMan}$.
  \item
    $*\text{-}\mathfrak{Alg}$: the category whose objects are topological unital $*$-algebras over $\mathbb{C}$ and whose morphisms are continuous, injective $*$-homomorphisms.
    Again the composition of morphisms is that of smooth maps and the unit element is the identity map.
\end{itemize}

\begin{definition}
  A \textbf{locally covariant conformal quantum field theory} is a covariant functor $\mathscr{A} : \mathfrak{CMan} \to *\text{-}\mathfrak{Alg}$.
  The theory is called \textbf{causal} if for any two morphisms $\psi_1 : M_1 \to M$ and $\psi_2: M_2 \to M$ of $\mathfrak{CMan}$ such that $\psi_1(M_1)$ and $\psi_2(M_2)$ are causally separated
  \begin{equation*}
    \big[
      \mathscr{A}(\psi_1)(\mathscr{A}(M_1, g_1)),
      \mathscr{A}(\psi_2)(\mathscr{A}(M_2, g_2))
    \big] = 0.
  \end{equation*}
  Furthermore, the theory is said to satisfy the \textbf{time-slice axiom} if $\mathscr{A}(\psi)(\mathscr{A}(M, g)) = \mathscr{A}(\widetilde M, \tilde g)$ for all morphisms $\psi : M \to \widetilde M$ of $\mathfrak{CMan}$ such that $\psi(M)$ contains a Cauchy surface of $\widetilde{M}$.
\end{definition}

In order to verify whether the vector potential fits in the scheme depicted by the last definition, we have to make sure that it is possible to associate to the vector potential a suitable $*$-algebra of observables.
In the algebraic approach to quantum field theory on curved spacetimes, the standard procedure calls for identifying the so-called \emph{field algebra} as a quotient of the Borchers-Uhlmann algebra with respect to the ideal generated by the equations of motion and the canonical commutation relations.

We recall that the solutions of \eqref{Proca} have conformal weight $0$ and that proposition \ref{symplform} guarantees the existence of a weakly non-degenerate symplectic form for the space of classical solutions of \eqref{dynamics} if $(M, g)$ has either trivial first or second de Rham cohomology group.
Similarly to \cite{pinamonti:2009zn} we can write:

\begin{definition}\label{algebra}
  We call {\bf field algebra} of the vector potential on a globally hyperbolic spacetime $(M, g)$ with $H^1(M) = \{0\}$ or $H^2(M) = \{0\}$ the $*$-algebra defined as the quotient $\mathcal{A}(M) \doteq \mathcal{A}^0(M) / \mathcal{I}(M)$.
  Here, $\mathcal{A}^0(M)$ is the free unital $*$-algebra generated by $\Omega_{0,\delta}^{1}(M, \mathbb{C})$, \textit{i.e.},
  \begin{equation*}
    \mathcal{A}^0(M) \doteq \mathbb{C} \oplus \bigoplus_{n=1}^\infty \Omega_{0,\delta}^{1}(M, \mathbb{C})^{\otimes n}
  \end{equation*}
  and $\mathcal{I}(M)$ is the closed $*$-ideal of $\mathcal{A}(M)$ generated by the elements
  \begin{equation*}
    \delta d \Omega^{1}_0(M, \mathbb{C})
    \quad \text{and} \quad
    -i G(f \otimes h) \oplus 0 \oplus (f \otimes h - h \otimes f),
  \end{equation*}
  for all $f, h \in \Omega^{1}_{0,\delta}(M, \mathbb{C})$, where $G$ is the causal propagator of $\Box$ on $(M, g)$.
  The $*$-operation is defined by complex conjugation.
\end{definition}

Recalling proposition \ref{confvecpot2}, we obtain the sought result namely that the field algebra yields a locally covariant conformal quantum field.

\begin{proposition}\label{lccqft}
  The vector potential can be described as a locally covariant conformal quantum field theory $\mathscr{A} : \mathfrak{CMan}' \to *\text{-}\mathfrak{Alg}$ which assigns to each object $(M, g)$ in $\mathfrak{CMan}'$ the $*$-algebra $\mathcal{A}(M)$ and to each conformal embedding $\psi : M \to \widetilde M$ with $\psi^* \tilde g = \Xi^2 g$ the $*$-homomorphism defined by $\alpha_\psi (q_M(a)) = q_{\widetilde M}(\psi_{(-2)\;*} a)$, where $a \in \mathcal{A}^0(M)$ and $q_M: \mathcal{A}^0(M) \to \mathcal{A}(M)$ is the quotient map.
  $\mathscr{A}$ is causal and satisfies the time-slice axiom.
\end{proposition}
\begin{proof}
  According to definition \ref{algebra} we can associate to $(M, g)$ the conformal field algebra $\mathcal{A}(M)$.
  Picking any arrow $\psi$ between $(M, g)$ and another object $(\widetilde M, \tilde g) \in \mathfrak{CMan}'$, we see that it is a conformal isometry from $(M, g)$ into $(\psi(M), \tilde g\!\restriction_{\psi(M)})$ which thus is also a proper map.
  In other words, $\Omega^{1}_{0}(M)$ is isomorphic to $\Omega^{1}_{0}(\psi(M))$.
  Since the codifferential acting on $1$-forms is invariant under the push-forward of conformal weight $-2$ -- see lemma \ref{new-lemma}, we also have an isomorphism $\Omega^{1}_{0,\delta}(M) \cong \Omega^{1}_{0,\delta}(\psi(M))$ via $\psi_{(-2)\;*}$.
  This entails that $\psi_{(-2)\;*} \mathcal{A}^0(M) = \mathcal{A}^0(\psi(M))$, \textit{i.e.}, given for example $a = a_0 \oplus a_1 \oplus (a_{2,1} \otimes a_{2,2}) \oplus \dotsb \in \mathcal{A}^0(M)$, we obtain $\psi_{(-2)\;*} a = a_0 \oplus \psi_{(-2)\;*} a_1 \oplus (\psi_{(-2)\;*} a_{2,1} \otimes \psi_{(-2)\;*} a_{2,2}) \oplus \dotsb \in \mathcal{A}^0(\psi(M))$.
  Furthermore, since $\psi_{(-2)\;*} \delta d f = \tilde \delta d \psi_* f$ for all $f \in \Omega^1_0(M)$ by \eqref{deltad} and since the symplectic form is invariant under conformal transformations by corollary \ref{confvecpot3}, one sees that $\psi_{(-2)\;*} \mathcal{I}(M) = \mathcal{I}(\psi(M))$.

  In order to understand that the map $\psi_{(-2)\;*}$ is also a $*$-homomorphism, notice that all operations involve only real structures and thus all complex conjugations are left untouched by $\psi_{(-2)\;*}$.
  To extend the result to $\widetilde M$ let us remember that $\psi(M) \subset \widetilde M$ and, thus, since all sections involved are compactly supported, we can extend each element of $\Omega^1_0(\psi(M))$ to $0$ on all points of $\widetilde M$ not lying in $\psi(M)$.
  Consequently $\mathcal{A}^0(\psi(M)) \subset \mathcal{A}^0(\widetilde M)$ and $\mathcal{I}(\psi(M)) \subset \mathcal{I}(\widetilde M)$, \textit{i.e.}, $\alpha_\psi \mathcal{A}(M) = \mathcal{A}^0(\psi(M)) / \mathcal{I}(\widetilde M) \subset \mathcal{A}(\widetilde M)$.
  Moreover, there does not exist $f \in \Omega^1_{0,\delta}(M)$ such that $f \notin \delta d \Omega^1_0(M)$ and $\psi_{(-2)\;*} f \in \tilde \delta d \Omega^1_0(\widetilde M)$ as this would imply that $d \tilde G \psi_{(-2)\;*} f = \psi_* d G f = 0$ while $d G f \neq 0$.
  To summarize, we have shown that $\alpha_\psi$ is a continuous (projections and push-forwards are continuous), injective $*$-homomorphism.

  Causality is a direct by-product of the definition of the ideal $\mathcal{I}(M)$ via the causal propagator $G$.
  As a matter of fact, if we consider two arrows $\psi_i: M_i \to M$ with causally separated images and two elements $f \equiv 0 \oplus f \oplus 0 \oplus \dotsb \in \mathcal{A}(\psi(M_1)) \subset \mathcal{A}(M)$ and $h \equiv 0 \oplus h \oplus 0 \oplus \dotsb \in \mathcal{A}(\psi(M_2)) \subset \mathcal{A}(M)$, then $[f, h] = f \otimes h - h \otimes f = i G(f \otimes h) = 0$ on account of the support properties of the causal propagator.
  Since the chosen elements are generating $\mathcal{A}(\psi(M_1))$ and $\mathcal{A}(\psi(M_2))$ respectively, this suffices to draw the sought conclusion.

  The time-slice axiom can be proven using a procedure similar to the one employed in \cite{dappiaggi:2011rt}:
  Let us pick any arrow $\psi: M\hookrightarrow \widetilde M$ between two objects such that $\psi(M)$ contains a Cauchy surface of $\widetilde M$.
  Consider now any $f \in \Omega^{1}_{0,\delta}(\widetilde M)$ and choose $\chi \in C^\infty(\widetilde M)$ such that $\chi = 1$ in $J^+(\psi(M); \widetilde M) \setminus \psi(M)$ and $\chi = 0$ in $J^-(\psi(M); \widetilde M) \setminus \psi(M)$.
  If we define $h \doteq f - \delta d(G^- f + \chi G f) = - \delta d (\chi G f)$, we have identified a compactly supported $1$-form $h$ in $\psi(M)$ which differs from $f$ by an element in $\mathcal{I}(\widetilde M)$.
  Since $\Omega^{1}_{0,\delta}(\widetilde M)$ generates $\mathcal{A}^0(\widetilde M)$, this suffices to conclude the proof.
\end{proof}

Note that the restrictions on the topology are not used anywhere in the last proof.
Indeed, had we chosen to define the field algebra for the vector potential in exactly the same way as we did in definition \ref{algebra} also for non-trivial $H^1(M)$ and $H^2(M)$, we would have also arrived at a locally conformally covariant formulation.
Such a choice, however, leads to a quantum theory which does not agree with the classical theory based on the field strength tensor as defined at the beginning of section \ref{sub:dynamics}.
As a matter of fact, as we commented already, the source of the problem lies in our choice according to which two vector potentials are equivalent if they differ by a closed $1$-form.
If we would have stuck with the traditional approach to identify two configurations differing by exact $1$-forms, the symplectic form would have been well-defined and weakly non-degenerate.
From a physical point of view the difference between these two perspectives is related to the Aharonov-Bohm observables, an effect which admittedly we are discarding with our choices (for a recent analysis of this issue, refer to \cite{Sanders:2012sf}).

If one wants to keep working with the former approach, a potential way out to construct a well-defined field algebra would be to mimic the approach of \cite{dappiaggi:2011yq,lang:2010hy} for the field strength tensor according to which the topological obstructions could be circumvented by associating to the underlying physical system the so-called universal algebra.
We will not give further details here, but we want to stress that, although this approach is feasible, there is a heavy price to pay: As shown in \cite{dappiaggi:2011yq,lang:2010hy} for the case of the field strength, the theory cannot be described within the framework of general local covariance.

On the contrary, if one would work with the more traditional approach, a very elegant way to quantize Maxwell's equations consists of following the BRST scheme which can be translated to the algebraic language as shown even in a broader scenario in \cite{hollands:2008zr} and more recently in \cite{fredenhagen:2011fk}.

\section{Hadamard States}\label{sec:hadamard}

This section will be fully devoted to the analysis of the second step in the algebraic quantization scheme, namely the choice of a suitable (algebraic) state.
Let us recall that this is a continuous linear functional $\omega : \mathcal{A}(M) \to \mathbb{C}$ such that $\omega(e) = 1$ and $\omega(a^*a) \geq 0$, where $e$ is the algebra unit while $a$ is an arbitrary element of $\mathcal{A}(M)$.
Notice that here $\mathcal{A}(M)$ could stand for an arbitrary unital topological $*$-algebra and not necessary for the one from definition \ref{algebra}.
The GNS theorem entails that the choice of $\omega$ is tantamount to the identification of a triple $(\mathcal{D}_\omega, \pi_\omega, \Omega_\omega)$ where $\mathcal{D}_\omega$ is a dense subset of an Hilbert space $\mathcal{H}_\omega$ and $\pi_\omega : \mathcal{A}(M) \to \mathcal{L}(\mathcal{D}_\omega)$ is a representation of the algebra in terms of linear operators acting on $\mathcal{D}_\omega$.
Moreover, $\Omega_\omega$ is a norm $1$ vector in $\mathcal{D}_\omega$ such that $\mathcal{D}_\omega = \{ \pi_\omega(a) \Omega,\; \forall a \in \mathcal{A}(M) \}$.
It also holds that, for all $a\in\mathcal{A}(M)$, the operators $\pi_\omega(a)$ are closable and $\omega(a)=\left(\Omega_\omega,\pi_\omega(a)\Omega_\omega\right)$.
From now on we will further restrict our attention to the subclass of possible states which are quasi-free, \textit{i.e.}, all $n$-point correlation functions $\omega_n$ can be constructed out of the two-point one, $\omega_2$, that is $\omega_n = 0$ for $n$ odd, whereas if $n$ is even
\begin{equation}\label{n-pt}
  \omega_n(f_1 \otimes \dotsm \otimes f_n)
  = \sum\limits_{\pi_n \in S_n} \prod\limits_{i=1}^{n/2} \omega_2(f_{\pi_n(2i-1)} \otimes f_{\pi_n(2i)}),
\end{equation}
where $f_i \in \mathcal{A}(M)$, $i = 1, 2, \dotsc, n$ and where $S_n$ denotes the ordered permutations of $n$ elements, namely $\pi_n$ fulfils that $\pi_n(2i-1) < \pi_n(2i)$ and $\pi_n(2i-1) < \pi_n(2i+1)$ for all $1 \leq i \leq n/2$.

The main goal of our investigation is to develop a procedure to identify in between the plethora of possible quasi-free states those which are physically meaningful.
On curved backgrounds it is nowadays widely accepted that such assertion coincides with asking that $\omega$ must be of Hadamard form, \textit{i.e.}, it must satisfy a condition on the singular structure\footnote{The identification of the Hadamard form with a specific condition on the structure of the two-point function is valid regardless of the assumption that $\omega$ is quasi-free as first proven in \cite{sanders:2010gf}.} of $\omega_2$.
Although from a mathematical point of view \emph{Hadamard states} can be characterised via the sophisticated tools of microlocal analysis as we shall discuss later, they are also noteworthy from a physical point of view.
As a matter of fact their ultraviolet behaviour mimics that of the Minkowski vacuum and the quantum fluctuations of all observables, such as, for example, the smeared components of the stress-energy tensor are bounded. These two conditions which are necessary for any state to be reasonably called physically sensible.
In order to treat interactions perturbatively, the algebra of fields has to be enlarged to include also the so-called Wick polynomials.
It was shown in \cite{hollands:2001ce} that a locally covariant notion of Wick products of fields can be defined for scalar fields by exploiting Hadamard states.
The reason is two-fold: On the one hand, in every normal neighbourhood, the singular structure of the two-point function of an Hadamard state depends only on the geometry of the underlying background and, hence it is universal.
On the other hand, the Hadamard condition is a statement concerning the singular structure of the bidistribution $\lambda_2$ naturally associated to $\omega_2$.
This allows us ultimately to use the H\"ormander criteria for the product of distributions.
Although the results of Hollands and Wald have been stated for scalar fields and later extended to Dirac fields, we expect that the very same procedure holds true also for the vector potential.

From a more practical perspective the study of Hadamard states is commonly divided in two distinct problems: existence and explicit construction.
If we focus on the specific case of the vector potential and its field algebra, as per definition \ref{algebra}, the first of these two problems has been already partly tackled: As first proven in \cite{fewster:2003or} employing a deformation technique first introduced in \cite{fulling:1981fm} and under the additional hypothesis that the underlying background $(M, g)$ is a globally hyperbolic spacetime with compact Cauchy surface $\Sigma$, Hadamard states do exist.
Yet, even though it is easily conceivable that such a result could be extended by removing the topological assumption on $\Sigma$, the deformation procedure cannot yield a mean to concretely construct a Hadamard state on non-trivial backgrounds.
This is instead the aim of this section and, as anticipated in the introduction, we shall achieve our goal on those asymptotically flat spacetimes $(M, g)$ which are genuine objects of $\mathfrak{CMan}'$.
The tool we shall use is the bulk-to-boundary reconstruction technique which we will now develop in detail for the vector potential and which has already been successfully applied to scalar and Dirac fields on asymptotically flat and cosmological spacetimes \cite{dappiaggi:2011ya,dappiaggi:2006cg,dappiaggi:2009th,dappiaggi:2009fk} as well as on Schwarzschild spacetime to construct the Unruh state \cite{dappiaggi:2011kx}.

\subsection{The Projection to the Boundary}\label{sub:bulk2boundary}

On account of the analysis of subsection \ref{sub:geometry}, every asymptotically flat spacetime $(M, g)$ comes endowed with a natural, codimension $1$, null, differentiable boundary $\mathscr{I}^+$ which has the property of being intrinsic and universal.
This suggests that we should use null infinity as the screen on which to encode the information of a bulk field theory and that a boundary theory has to be defined in such a way that it does not depend on the chosen bulk $(M, g)$.
Henceforth, let $(M, g)$ be an asymptotically flat spacetime as per subsection \ref{sub:asymp_flat}.
We shall further assume that either the first or the second de Rham cohomology group of $M$ is trivial so that the hypotheses of proposition \ref{symplform} are met and we can associate to the vector potential on $M$ a well-defined field algebra.
It is worth mentioning that the topological constraints imposed on $M$ might be automatically satisfied by all asymptotically flat spacetimes\footnote{We are grateful to the referee for pointing out this possibility.}.
This conjecture would require a detailed analysis which is beyond the scope of this paper.
Let us call $\psi$ the embedding into the unphysical spacetime $(\widetilde M, \tilde g)$ and let $\iota : \mathscr{I}^+ \hookrightarrow \widetilde M$ be the smooth embedding of null infinity into the unphysical spacetime such that $g_B = \xi_B^2 \iota^* \tilde g$ is of Bondi form and $\xi_B^{-1} \iota^* n$ is complete.
Here $\xi_B$ is the very same function introduced in section \ref{sub:geometry}.
In the following we shall only work in this particular frame which induces on null infinity a natural measure $\mu_\mathscr{I} \doteq \sin^2\!\theta\, du\, d\theta\, d\varphi$.

The fibers of the pull-back bundle $\iota^* T \widetilde M$ can be decomposed at every point $p \in \mathscr{I}^+$ into the subspace $N_p$ spanned by the vectors $\partial_u$ and $\partial_\Xi$, and the subspace $S_p$ spanned by $\partial_\theta$ and $\partial_\varphi$.
$g_B$ can then be split as well, to induce on $S_p$ the metric tensor of the unit $2$-sphere which we will denote by $g_S$.
Thus we obtain the trivial vector bundle $S \to \mathscr{I}^+$ with metric tensor $g_S$ and it's dual $S^*$ equipped with the inverse metric tensor.
These bundles are easily defined intrinsically on $\mathscr{I}^+$ as long as we remember that under a BMS transformation $\gamma \in BMS$ we have the scaling behaviour $g_S \mapsto K_\gamma^2 g_S$ induced by the metric tensor on $\mathscr{I}^+$ itself.
Moreover we can define a projection $\pi: \iota^* T^* \widetilde M \hookrightarrow S^*$ from $\iota^* T^* \widetilde M$ into the subbundle $S^*$.

Together with the measure $\mu_\mathscr{I}$, $g_S$ allows the definition of a (non-negative) pairing
\begin{equation*}
  \langle f, h \rangle \doteq \int\limits_{\mathscr{I}^+} g_S^{-1}(f, h)\, \mu_{\mathscr{I}}
\end{equation*}
for all sections $f, h' \in \Gamma(S^*)$ for which the integral converges.
Using this pairing we can now construct a symplectic space intrinsically on $\mathscr{I}^+$:

\begin{proposition}
  The set
  \begin{equation}\label{boundspace}
    \mathcal{S}(\mathscr{I}) \doteq \{ f \in \Gamma(S^*) \mid \langle f, f \rangle < \infty \textrm{ and } \langle \partial_u f, \partial_u f \rangle < \infty \}
  \end{equation}
  is a symplectic space invariant under the representation $\Pi^0$ of the BMS group (cf. \eqref{BMSrep}) if endowed with the following weakly non-degenerate symplectic form $\varsigma : \mathcal{S}(\mathscr{I}) \times \mathcal{S}(\mathscr{I}) \to \mathbb{R}$ such that
  \begin{equation}\label{boundsympl}
    \varsigma(f, h) \doteq \langle f, \partial_u h \rangle - \langle \partial_u f, h \rangle = 2\, \langle f, \partial_u h \rangle
  \end{equation}
  for any $f, h \in \mathcal{S}(\mathscr{I})$.
\end{proposition}
\begin{proof}
  As a first step we check that the right-hand side of \eqref{boundsympl} is well-defined.
  Using the Cauchy-Schwarz inequality, we see that the right-hand side of \eqref{boundsympl} is well-defined by definition of $\mathcal{S}(\mathscr{I})$.
  Since $(f, f)_{\mathscr{I}}$ with $f \in \mathcal{S}(\mathscr{I})$ is integrable, it turns out that $\lim_{u \to \pm \infty} f = 0$ as one can establish with a minor readaptation of the proof in \cite[Footnote 7]{moretti:2008ay}.
  Therefore we can employ integration by parts to obtain the equality
  \begin{equation*}
    \langle f, \partial_u h \rangle - \langle \partial_u f, h \rangle = 2\, \langle f, \partial_u h \rangle.
  \end{equation*}
  Moreover we see by direct inspection that the symplectic form is both bilinear and antisymmetric.
  The non-degenerateness of $\varsigma(\cdot\,, \cdot)$ follows from the non-degenerateness of $g_S$ and the trivial kernel of $\partial_u$ acting on $\mathcal{S}(\mathscr{I})$ (which is a consequence of the condition $\langle f, f \rangle < \infty$).

  Concerning BMS invariance, let us consider a generic element $\gamma \in BMS$.
  According to the translation invariance of the measure, it holds
  \begin{equation*}
    \varsigma(f, h)
    \mapsto 2 \int\limits_{\mathscr{I}^+} K_\gamma^{-2}\, g_S^{-1} (f, K_\gamma^{-1} \partial_u h) \, K_\gamma^{3} \, \mu_{\mathscr{I}}
    = \varsigma(f, h),
  \end{equation*}
  which is indeed the sought property.
\end{proof}

We can employ this symplectic space to associate a $*$-algebra to null infinity:

\begin{definition}
  We call \textbf{boundary algebra} $\mathcal{A}(\mathscr{I})$ the $*$-algebra realised as the quotient between the free unital $*$-algebra $\mathcal{A}(\mathscr{I})^0 \doteq \mathbb{C} \oplus \bigoplus_{n=1}^\infty \mathcal{S}(\mathscr{I})^{\otimes n} \otimes \mathbb{C}$ and the $*$-ideal generated by the elements of the form $-i \varsigma(f,h) \oplus 0 \oplus (f \otimes h - h \otimes f)$ for all $f, h \in \mathcal{S}(\mathscr{I})$.
  As usual, the $*$-operation is defined in terms of complex conjugation.
\end{definition}

Having established an abstract boundary theory, we will now bring to attention a striking relationship between the Maxwell theory in the bulk and the theory on null infinity.
To this avail, let us first individuate a special gauge on null infinity:

\begin{proposition}\label{gaugefix}
  For each $[f] \in \mathcal{S}(M)$ there exists a unique representative $f \in [f]$ such that $A = \tilde G \psi_{(-2)\;*} f$ satisfies $(\iota^* A)_u = 0$.
  We say that $f$ is in the $\mathscr{I}$-gauge.
\end{proposition}
\begin{proof}
  Let $f' \in [f]$ be arbitrary and set $A' = \tilde G \psi_{(-2)\;*} f'$.
  Integrating $\partial_n \chi = - A'_n \restriction_{\mathscr{I}^+ \cup i^+}$ and choosing the constant of integration to be zero, we obtain a smooth function $\chi \in C^\infty(\mathscr{I}^+ \cup i^+)$ such that $\iota^* \chi(u, \theta, \varphi) = 0$ for all $u < u'$ fixed.
  We may then use $\chi$ as initial data for the characteristic initial value problem $\tilde \Box \lambda = \tilde \delta d \lambda = 0$, $\lambda \restriction_{\mathscr{I}^+} = \chi$ to find a gauge function $\lambda \in C^\infty(\overline{\psi(M)})$ such that $(\iota^* A' + \iota^* d\lambda)_u = 0$.
  Hence, by the uniqueness of the solution to both the characteristic initial value problem and the Cauchy problem, and corollary \ref{cohomology}, there exists $f \in [f]$ as specified in the thesis.
\end{proof}

Note again that this residual gauge transformation fixes the gauge within $\mathcal{S}(M)$ completely because any further gauge transformation would necessarily be constant along the $u$-direction, compactly supported at $\mathscr{I}^+$ and thus identically zero by the uniqueness of the solution of the characteristic initial value problem.

Let us now define a map $b : \mathcal{S}(M) \to \mathcal{S}(\mathscr{I})$ which we will show to be a symplectomorphism from the bulk theory $(\mathcal{S}(M), \sigma)$ to the boundary theory $(\mathcal{S}(\mathscr{I}), \varsigma)$ in the next theorem.

\begin{theorem}\label{main}
  Let $b : \mathcal{S}(M) \to \mathcal{S}(\mathscr{I})$ be the map defined by
  \begin{equation*}
    b([f]) = \pi (\iota^* (\tilde G^- \psi_{(-2)\;*} f)),
  \end{equation*}
  where $f \in [f]$ is in the $\mathscr{I}$-gauge.
  Then $b$ is a symplectomorphism, \textit{i.e.}, for all $[f],[h] \in \mathcal{S}(M)$ it holds that
  \begin{equation*}
    \sigma([f],[h]) = \varsigma(b([f]), b([h])),
  \end{equation*}
  where $\sigma$ is given in \eqref{symplectic} and $\varsigma$ in \eqref{boundsympl}.
  $b$ is called the \textbf{bulk-to-boundary projection}.
\end{theorem}
\begin{proof}
  First, we show that $b([f]), \partial_u b([f]) \in \mathcal{S}(\mathscr{I})$.
  $A = \tilde G^- \psi_{(-2)\;*} f$ is compactly supported on $\mathscr{I}^+ \cup i^+$.
  Therefore the surface integral of $\tilde g^{-1}(A, A)$ on $\mathscr{I}^+$ converges and we have
  \begin{equation*}
    \infty
    > \int\limits_{\mathscr{I}^+} \tilde *\, \tilde g^{-1}(A, A)\, d\Xi
    = \int\limits_{\mathscr{I}^+} \iota^* \big( \tilde *\, \tilde g^{-1}(A, A)\, d\Xi \big)
    = \int\limits_{\mathscr{I}^+} g_B^{-1}(\iota^* A, \iota^* A)\, \mu_{\mathscr{I}}
    = \langle \pi (\iota^* A), \pi (\iota^* A) \rangle.
  \end{equation*}
  Here we used in first step that $i^+$ has zero measure and that the integral is invariant under the pull-back via $\iota$ ($\iota$ being an embedding of $\mathscr{I}^+$ into $\widetilde M$ is also a diffeomorphism of $\mathscr{I}^+$ into itself).
  The second equality holds because the conformal weights of the measure and the inverse metric cancel and in the last step we used the vanishing of the $u$-component $(\iota^* A)_u = 0$.\footnote{Remember that $g_B^{-1}$ contracts the $u$- and the $\Xi$-component of $\iota^* A$.}
  Since $\xi_B^{-1}$ falls off as we approach $i^+$ \cite[Lemma 4.4]{moretti:2006ap}, we obtain analogously
  \begin{equation*}
    \infty
    > \int\limits_{\mathscr{I}^+} \tilde *\, (\xi_B^{-1} \partial_n A \wedge \tilde *\, \xi_B^{-1} \partial_n A)\, d\Xi
    = \int\limits_{\mathscr{I}^+} g_B^{-1}(\partial_u \iota^* A, \partial_u \iota^* A)\, \mu_{\mathscr{I}}
    = \langle \partial_u \pi (\iota^* A), \partial_u \pi (\iota^* A) \rangle,
  \end{equation*}
  where the $\xi_B^{-1}$ factors cancel with the factor from $\iota^* \partial_n = \xi_B \partial_u$.
  Hence $b([f])$ and $\partial_u b([f])$ satisfy the conditions specified for $\mathcal{S}(\mathscr{I})$ in \eqref{boundspace}.

  Next, we show that $\sigma([f],[h]) = \varsigma(b([f]), b([h]))$.
  By \ref{confvecpot3} we have
  \begin{equation*}
    \sigma([f],[h])
    = \int\limits_{\tilde M} \tilde G \psi_{(-2)\;*} f \wedge \tilde *\, \psi_{(-2)\;*} h
    = \int\limits_{\tilde M} \tilde G \psi_{(-2)\;*} f \wedge \tilde *\, \tilde \delta d \tilde G^- \psi_{(-2)\;*} h,
  \end{equation*}
  where $f, h \in \Omega^1_{0, \delta}(M)$ are chosen to be in the $\mathscr{I}$-gauge.
  Applying Stokes' theorem, we then obtain
  \begin{align*}
    \sigma([f],[h])
    & = \int\limits_{\mathscr{I}^+} \iota^* \big( \tilde G^- \psi_{(-2)\;*} h \wedge \tilde * \, d \tilde G^- \psi_{(-2)\;*} f - \tilde G^- \psi_{(-2)\;*} f \wedge \tilde * \, d \tilde G^- \psi_{(-2)\;*} h \big)
    \\
    & = \int\limits_{\mathscr{I}^+} \big( g_B^{-1}(\iota^* \tilde G^- \psi_{(-2)\;*} f, \partial_u \iota^* \tilde G^- \psi_{(-2)\;*} h) - g_B^{-1}( \partial_u \iota^* \tilde G^- \psi_{(-2)\;*} f, \iota^* \tilde G^- \psi_{(-2)\;*} h) \big) \mu_{\mathscr{I}}
    \\
    & = \langle b([f]), \partial_u b([h]) \rangle - \langle \partial_u b([f]), b([h]) \rangle,
  \end{align*}
  where we have used again that $i^+$ has zero measure, that the conformal weights of the measure and the inverse metric cancel, and that $(\iota^* \tilde G^- \psi_{(-2)\;*} f)_u = (\iota^* \tilde G^- \psi_{(-2)\;*} h)_u = 0$.
\end{proof}

We have thus proven that $b$ is a symplectomorphism between $(\mathcal{S}(M), \sigma)$ and $(\mathcal{S}(\mathscr{I}), \varsigma)$ which, analogously to what happened in \cite{dappiaggi:2006cg} for the Weyl algebra, induces a natural injective $*$-homomorphism between the bulk and the boundary algebra.
Denoting this $*$-homomorphism by the same symbol, we thus obtain:

\begin{lemma}\label{b2b}
  The symplectomorphism $b$ induces an injective $*$-algebra homomorphism $b : \mathcal{A}(M) \to \mathcal{A}(\mathscr{I}^+)$ such that
  \begin{equation*}
    b(A_M(f)) \doteq A_M(b([f])),
  \end{equation*}
  for all $[f] \in \mathcal{S}(M)$.
  Accordingly, every algebraic state $\omega^{\mathscr{I}}: \mathcal{A}(\mathscr{I}^+) \to \mathbb{C}$ on the boundary algebra induces a state $\omega^M : \mathcal{A}(M) \to \mathbb{C}$ on the bulk field algebra unambiguously defined by $\omega^M \doteq \omega^{\mathscr{I}} \circ b$.
\end{lemma}

\subsection{Hadamard States Induced from Null Infinity}\label{sub:state}

As we have already remarked at the beginning of this section, our goal is to provide a scheme to construct Hadamard states on a non-trivial background $(M, g)$.
The traditional techniques rely on $g$ possessing a complete timelike Killing field, that is $M$ is stationary, so that it is possible to make sense of the notion of positive energy out of a Fourier transform.
Clearly, this procedure fails in the most general scenario since the group of isometries of $(M, g)$ can even be trivial.
The last lemma of the previous subsection provides us with a potential way to circumvent such an obstruction by inducing a state for the bulk field theory out of one for the boundary counterpart.
This point of view has an {\it a priori} advantage due to the geometry of $\mathscr{I}^+$ which possesses a natural direction, identified by the $u$-coordinate in a Bondi frame, along which one can perform a Fourier transform.
This direction plays for the boundary the same role as the time direction in a stationary or static spacetime, thus allowing us to single out a natural and distinguished algebraic state for the boundary theory \cite{dappiaggi:2006cg,moretti:2006ap,moretti:2008ay}.
We will now show that this line of reasoning can also be applied to the case of the vector potential.
Notice that a more detailed account of the analysis presented here is also available in \cite{siemssen:2011fk}.
As a starting point we focus on the boundary $*$-algebra:

\begin{proposition}\label{bulkstate}
  The map $\omega_2^{\mathscr{I}} : \mathcal{S}(\mathscr{I}) \otimes \mathcal{S}(\mathscr{I}) \to \mathbb{R}$ such that
  \begin{equation}\label{2ptScri}
    \omega_2^{\mathscr{I}}(f \otimes h) \doteq \lim\limits_{\epsilon \to 0} - \frac{1}{\pi} \int\limits_{\mathbb{R}^2 \times \mathbb{S}^2} \frac{\big( f(u,\theta,\varphi), h(u',\theta,\varphi) \big)_{\mathscr{I}}}{(u - u' - i \epsilon)^2}\, du\, du'\, d\mathbb{S}^2(\theta,\varphi),
\end{equation}
  where $d\mathbb{S}^2(\theta,\varphi) = \sin^2\! \theta \, d\theta \, d\varphi$ is the standard measure on the unit $2$-sphere, unambiguously identifies a quasi-free state $\omega^{\mathscr{I}} : \mathcal{A}(\mathscr{I}) \to \mathbb{C}$.
  Furthermore,
  \begin{enumerate}
    \item
      $\omega^{\mathscr{I}}$ induces a bulk state $\omega^M \doteq \omega^{\mathscr{I}} \circ b : \mathcal{A}(M) \to \mathbb{C}$ which is quasi-free and
    \item
      $\omega^{\mathscr{I}}$ is invariant under the action of the BMS group $\rho : \mathcal{A}(\mathscr{I}) \to \mathcal{A}(\mathscr{I})$ induced by \eqref{BMS}, \textit{viz.} $\omega^{\mathscr{I}} \circ \rho = \omega^{\mathscr{I}}$.
  \end{enumerate}
\end{proposition}
\begin{proof}
  Since the elements $f$ of $\mathcal{S}(\mathscr{I})$ are such that $\langle f, f \rangle_{\mathscr{I}} < \infty$ and $\langle \partial_u f, \partial_u f \rangle_{\mathscr{I}} < \infty$, barring a minor readjustment, we can employ the analysis of \cite{dappiaggi:2009th,dappiaggi:2009fk} to conclude that $\omega_2^{\mathscr{I}}$ identifies a well-defined two-point distribution on null infinity.
  Because we know that both $f$ and $h$ tend to $0$ as $u$ diverges, we can apply integration by parts, Parseval's identity and the convolution theorem\footnote{The conventions for the Fourier transform are as in \cite{hormander:1990bl}.} to obtain
  \begin{equation}\label{2ptmodes}
    \omega^{\mathscr{I}}_2(f \otimes h)
    = \frac{1}{\pi} \int\limits_{\mathbb{R} \times \mathbb{S}^2} k \, \Theta(k) \big(\widehat{f}(k,\theta,\varphi), \widehat{h}(-k,\theta,\varphi) \big)_{\mathscr{I}}\, dk \, d\mathbb{S}^2(\theta,\varphi),
  \end{equation}
  where $\Theta(k)$ is the Heaviside step function.
  Hence $\omega^{\mathscr{I}}_2(f \otimes \overline{f}) \geq 0$ and the corresponding state $\omega^{\mathscr{I}}$ satisfies positivity.
  A direct calculation shows that $\omega^{\mathscr{I}}_2(f \otimes h) - \omega^{\mathscr{I}}_2(h \otimes f) = i \varsigma(f, h)$ and thus $\omega^{\mathscr{I}}_2$ unambiguously identifies an algebraic quasi-free state for $\mathcal{A}(\mathscr{I})$.

  Concerning BMS invariance, it is sufficient to prove that it holds true for the two-point function as the state is quasi-free.
  Therefore the natural representation of a BMS transformation $\gamma \in BMS$ on $\mathcal{S}(\mathscr{I})$ via $\Pi^{0}_\gamma$ yields
  \begin{equation*}
    (\omega^{\mathscr{I}}_2 \circ \rho)(f \otimes h)
    = \lim\limits_{\epsilon \to 0} - \frac{1}{\pi} \int\limits_{\mathbb{R}^2 \times \mathbb{S}^2} \frac{K_\gamma^{-2} \big( K_\gamma^{0} f(u,\theta,\varphi), K_\gamma^{0} h(u',\theta,\varphi) \big)_{\mathscr{I}}}{(K_\gamma u - K_\gamma u' - i \epsilon)^2} \, K_\gamma^4 \, du \, du' \, d\mathbb{S}^2(\theta,\varphi),
  \end{equation*}
  where we suppressed the explicit dependence of the conformal factor $K_\gamma(\theta,\varphi)$ on the angular variables.
  Furthermore, since this factor is bounded, the right hand side coincides with $\omega^{\mathscr{I}}_2(f \otimes h)$ up to an irrelevant redefinition of the $\epsilon$-parameter.
\end{proof}

Although we have identified a BMS-invariant state on $\mathscr{I}^+$, we are still far from claiming that its bulk counterpart is physically meaningful.
In order to answer this question, first of all we need to make precise the concept of Hadamard states.
Suppose for now that we are equipped with an arbitrary but quasi-free algebraic state $\lambda : \mathcal{A}(M) \to \mathbb{C}$ with associated two-point function $\lambda_2(f \otimes h)$ for all $f, h \in \Omega^1_{0,\delta}(M)$.
To summarize, it turns out that the structure of $\lambda_2$ is best studied by means of the techniques proper of microlocal analysis.
These are presented in the monograph \cite{hormander:1990bl} whose notations, definitions and conventions we will adopt here.
Since in our scenario we cope with vector bundles instead of scalar functions, we need to slightly extend the ordinary notion of the wavefront set.
As noted after \cite[Th. 8.2.4]{hormander:1990bl}, the wavefront set $\mathrm{WF}(u)$ of a distribution $u \in \mathcal{D}'(E)$ on a vector bundle $E$ is defined locally as $\cup_i \mathrm{WF}(u_i)$, where $(u_1, \cdots, u_N)$ are the components of $u$ in a local trivialisation of $E$.
This definition is indeed invariant under a change of the local trivialisation and thus yields a sensible extension of the wavefront set to distributional sections.
Hence we can follow \cite{radzikowski:1996ug,radzikowski:1996ul,sahlmann:2001sr}:

\begin{definition}\label{WF}
  We say that a two-point distribution $\lambda_2$ is of \textbf{Hadamard form} if it satisfies
  \begin{equation*}
    \mathrm{WF}(\lambda_2) = \{ (x, x', k, -k') \in T^*M^2 \setminus 0 \mid (x,k) \sim (x,k'),\, k \triangleright 0 \},
  \end{equation*}
  where $0$ is the zero section of $T^*M^2$.
  Here $(x,k) \sim (x',k')$ means that the point $x$ and $x'$ are connected by a null geodesic $\gamma$ so that $k$ is cotangent to it in $x$ and $k'$ is the parallel transport along $\gamma$ of $k$ from $x$ to $x'$.
  Furthermore, $k \triangleright 0$ means that the covector $k$ is future-directed.
\end{definition}

As noted in \cite{fewster:2003or}, a generic state $\lambda$ on $\mathcal{A}(M)$ does not yield a two-point function $\lambda_2$ that is a distribution because it acts only on coclosed test forms.
Instead one says that $\lambda$ is a Hadamard state if there exists a $\Box$-bisolution $\Lambda_2$ which is of Hadamard form such that $\lambda_2(f \otimes h) = \Lambda(f \otimes h)$ for all $f, h \in \Omega^{1}_{0,\delta}(M)$.
To this end, let us notice that $\omega^M_2$ is a $\Box$-bisolution
\begin{equation*}
  \omega^M_2(\Box f \otimes h) = 0 = \omega^M_2(f \otimes \Box h)
\end{equation*}
for all $f, h \in \Omega^{1}_{0}(M)$.
This equality holds in particular for coclosed $f, h$ so that $\omega^M_2$ is also bisolution of the equations of motion \eqref{dynamics}.
Thanks to this, we can prove the main theorem of this section:

\begin{theorem}
  The state $\omega^M \doteq \omega^{\mathscr{I}} \circ b : \mathcal{A}(M) \to \mathbb{C}$ where $\omega^{\mathscr{I}} : \mathcal{A}(\mathscr{I}) \to \mathbb{C}$ is the unique quasi-free state built out of \eqref{2ptScri}
  \begin{enumerate}
    \item
      is of Hadamard form,
    \item
      is invariant under the action of all isometries of $(M, g)$, that is $\omega^M \circ \alpha_\phi = \omega^M$ for all possible isometries $\phi$ of $(M, g)$.
      Here $\alpha$ is the $*$-isomorphism representation of the isometry group on $\mathcal{A}(M)$ defined by the action on the algebra generators $f \in \Omega^{1}_{0,\delta}$, \textit{i.e.}, $\alpha_\phi(f) \doteq \phi_* f$.
    \item
      coincides with the Poincar\'e vacuum if the bulk spacetime is Minkowski.
  \end{enumerate}
\end{theorem}
\begin{proof}\textit{a)\hspace{\labelsep}} Since $\beta_M$ is a quasi-free state, we focus only on $\omega^M_2$.
  In order to prove the Hadamard property, it is convenient to work with the auxiliary object
  \begin{equation*}
    \omega^{\widetilde M}_2(f \otimes h)
    \doteq \omega^{\mathscr{I}}_2( \iota^* \tilde G^- f \otimes \iota^* \tilde G^- h),
  \end{equation*}
  where $f, h \in \Omega^1_{0}(\widetilde M)$.
  Notice that this is a well-defined expression since $J^+(\supp(f + h); \widetilde M) \cap J^-_{\widetilde M}(i^+)$ is compact and that, on account of the support properties of the retarded fundamental solution of the Laplace-de Rham operator, $\supp \omega^{\widetilde M}_2 \subseteq J^-(i^+; \widetilde M) \setminus i^-$.
  Furthermore it holds
  \begin{equation*}
    \omega^{\widetilde M}_2(\tilde \Box f \otimes \tilde \Box h)
    = \omega^{\mathscr{I}}_2(\iota^* f \otimes \iota^* h).
  \end{equation*}
  Hence the wavefront set of $(\tilde \Box \otimes \tilde \Box) \, \omega^{\widetilde M}_2$ coincides with that of $\omega^{\mathscr{I}}_2$ which can be calculated directly from \eqref{2ptScri} and, following \cite{dappiaggi:2009fk,moretti:2008ay},
  \begin{equation*}
    \mathrm{WF}(\omega^{\mathscr{I}}_2) = \big\{ (x,x,k,-k) \in T^*(\mathscr{I}^+)^2 \setminus 0 \mid k_u > 0 \big\}.
  \end{equation*}
  If we apply the theorem of propagation of singularities, we obtain
  \begin{multline*}
    \mathrm{WF}(\omega^{\widetilde M}_2) = \big\{ (x,x',k,-k') \in T^*\widetilde{M}^2 \setminus 0 \mid \exists p \in \mathscr{I}^+,\, q \in T^*_p\widetilde{M} \textrm{ with } q_u > 0 \textrm{ such that }\\
    x, x' \in J^-_{\widetilde M}(i^+) \setminus i^-,\, (x,k) \sim (x,k') \sim (p,q) \big\}.
  \end{multline*}
  To conclude, we need to restrict our test functions to those which are compactly supported in $\psi(M)$, the image of $M$ in $\widetilde{M}$.
  Then, on account of the invariance of null geodesics under conformal transformations and of the fact that there does not exist a null geodesic joining $x \in \psi(M)$ with $i^+$ \cite[Lemma 4.3]{moretti:2008ay}, we can conclude that the wavefront set of $\omega^M_2$ has the same form as that in definition \ref{WF}.
  Thus $\omega_M$ is a Hadamard state.

  \vspace{\itemsep}
  \textit{b)\hspace{\labelsep}}The same argument given in \cite[Th. 3.1]{moretti:2008ay} guarantees that the statement holds true if proven for all one-parameter groups of isometries $\phi^X_t$, $t \in \mathbb{R}$ induced by a Killing vector field $X$.
  On account of proposition \ref{bulkstate}, $\omega^M \circ \alpha_{\phi^X_t} = \omega^{\mathscr{I}} \circ b \circ \alpha_{\phi^X_t}$.
  Furthermore it holds \cite[Prop. 3.4]{moretti:2008ay} that $b \circ \alpha_{\phi^X_t} = \rho_{\tilde{\phi}^{\tilde{X}}_t} \circ b$, where $\rho_{\tilde{\phi}^{\tilde{X}}_t}$ is the action on $\mathcal{A}(\mathscr{I})$ of a one-parameter subgroup of the BMS induced via the exponential map from $\tilde X$, the unique extension of $\tilde X$ to the boundary.
  Since in proposition \ref{2ptScri} we proved the invariance of $\omega^{\mathscr{I}}$ under the action of the BMS group,
  \begin{equation*}
    \omega^M \circ \alpha_{\phi^X_t} = \omega^{\mathscr{I}} \circ \rho_{\tilde{\phi}^{\tilde{X}}_t} \circ b = \omega^{\mathscr{I}} \circ b = \omega^M.
  \end{equation*}

  \vspace{\itemsep}
  \textit{c)\hspace{\labelsep}}Having already established both a) and b), we can conclude that, in Minkowski spacetime, $\omega^M$ is a quasi-free, Hadamard state which is Poincar\'e invariant.
  Either arguing via the uniqueness of the ground state in flat spacetime or repeating the proof of \cite[Th. 4.1]{dappiaggi:2006cg}, we establish c).
\end{proof}

An additional advantage of inducing states for the bulk field algebra from the boundary counterpart is the possibility to define a KMS state for the boundary theory on $\mathscr{I}^+$.
This was first introduced in \cite{dappiaggi:2011kx} to rigorously define the Unruh state on Schwarzschild spacetime and it was later applied to scalar and Dirac fields on a certain class of cosmological spacetimes in \cite{dappiaggi:2011ya}.
The rationale behind the whole procedure is that the boundary state can be constructed starting from a two-point function which admits an explicit mode decomposition and which can be easily modified by adding a Bose factor to construct a KMS state.
We will not dwell on the definition of such a class of states and we refer to \cite[Appendix B]{dappiaggi:2011kx} for an overview.
Hence, in the case at hand, the starting point would be \eqref{2ptmodes} out of which we define
\begin{equation}\label{2ptKMS}
  \omega^{\mathscr{I}}_{2,T}(f \otimes h)
  = \frac{1}{\pi} \int\limits_{\mathbb{R} \times \mathbb{S}^2} \frac{k}{1 - e^{-\frac{k}{T}}} \big( \widehat{f}(k,\theta,\varphi), \widehat{h}(-k,\theta,\varphi) \big)_{\mathscr{I}}\, dk\, d\mathbb{S}^2(\theta,\varphi),
\end{equation}
where $T \geq 0$ and $f, h \in \mathcal{S}(\mathscr{I})$.

\begin{proposition}
  The two-point function \eqref{2ptKMS} induces a state $\omega^{\mathscr{I}}_T : \mathcal{A}(\mathscr{I}) \to \mathbb{C}$ which is quasi-free, enjoys the KMS property with respect to the $u$-translations on $\mathscr{I}^+$ and which converges weakly to the state $\omega^{\mathscr{I}}$ of proposition \ref{bulkstate} as $T \to 0$.
\end{proposition}
\begin{proof}
  Since $0 \leq \omega^{\mathscr{I}}_2(f \otimes \overline{f}) < \omega^{\mathscr{I}}_{2,T}(f \otimes \overline{f})$, we can infer that a quasi-free state built out of \eqref{n-pt} is indeed positive.
  It is also rather straightforward to see that such a state satisfies the canonical commutation relations since, as we will explicitly write out in the next lemma, the difference between \eqref{2ptKMS} and the counterpart at $T=0$ is symmetric.
  Furthermore $\omega^{\mathscr{I}}_T$ is invariant under the translations generated by $\partial_u$, {\it i.e.}, $f(u,\theta,\varphi) \mapsto \alpha_{\lambda} f(u,\theta,\varphi) \doteq f(u-\lambda,\theta,\varphi)$.
Therefore $\widehat{f}(k,\theta,\varphi) \mapsto \widehat{f}(k,\theta,\varphi) e^{i k \lambda}$ for some $\lambda \in \mathbb{R}$ and any $f \in \mathcal{S}(\mathscr{I})$, as can be inferred by direct inspection of the explicit form of the two-point function \eqref{2ptKMS}.
  We are thus left with proving the KMS condition.
  Given $f, h \in \mathcal{S}(\mathscr{I})$ and taking the analytic continuation of $f, h$ in the $u$-variable to the complex plane, we have
  \begin{equation*}
    \omega^{\mathscr{I}}_{2,T}(f \otimes \alpha_{i T^{-1}} h) = \omega^{\mathscr{I}}_{2,T}(h \otimes f)
  \end{equation*}
  which guarantees that the KMS property holds true for $\omega^{\mathscr{I}}_T$ \cite[Sect. 5.3]{bratelli:1997nk}.
  The weak convergence of $\beta_T$ to $\beta$ as $T\to 0$ is an immediate consequence of the explicit form of \eqref{2ptKMS} and of the regularity properties of each $f,h\in\mathcal{S}(\mathscr{I})$
\end{proof}

If we follow lemma \ref{b2b}, we know that $\omega^{\mathscr{I}}_T$ induces a bulk state $\omega^M_T$ but, a priori, we cannot hope that it will be thermal unless $(M, g)$ is a static spacetime and hence endowed with a timelike Killing field with respect to which we can define a KMS condition.
In this case we can use the results of proposition 3.1 in \cite{moretti:2008ay} which guarantee that, if $(M,g)$ possesses a timelike Killing field, then it smoothly extends to a vector field of the unphysical spacetime whose restriction on null infinity generates a one-parameter subgroup of the supertranslations.
Furthermore, since the timelike Killing field can be chosen future-directed, a minor adaption of proposition 3.3 and 3.2 in \cite{moretti:2008ay} allows to conclude that there must exist an admissible frame such that the action of the one-parameter subgroup induced by the BMS generator associated to the bulk timelike Killing field reduces to a rigid translation along the null coordinate of $\mathscr{I}$.
Under these hypotheses, the construction of the boundary KMS state and of its bulk counterpart entails that the latter is a KMS state with respect to the timelike Killing field.
Furthermore, as already outlined in \cite{dappiaggi:2011ya}, even if the underlying background is not static, the bulk state has a reminiscence of the exact thermal structure on the boundary and thus one might use it as a natural candidate to discuss physical phenomena for which we would like to give an at least approximate or asymptotic thermal interpretation.
To this avail, we need to make sure that $\omega^M_T$ is physically meaningful and as a last task of this section we prove

\begin{proposition}
  The state $\omega^M_T \doteq \omega^{\mathscr{I}}_T \circ b : \mathcal{A}(M) \to \mathbb{C}$ is a quasi-free, Hadamard state.
\end{proposition}
\begin{proof}
  Since $\omega^{\mathscr{I}}_T$ is a quasi-free state, $\omega^M_T$ enjoys, per construction, the very same property.
  In order to prove that it is Hadamard, we can focus on the associated two-point function $\omega^M_{2,T}$ and we introduce as an auxiliary tool
  \begin{equation*}
    \Delta_T(f \otimes h) \doteq \omega^M_{2,T}(f \otimes h) - \omega^M_2(f \otimes h),
  \end{equation*}
  for all $f, h \in \Omega_{0,\delta}^1(M)$.
  In term of modes this last expression reads:
  \begin{equation*}
    \Delta_T(f,h) = \frac{1}{\pi} \int\limits_{\mathbb{R} \times \mathbb{S}^2} \frac{|k|}{e^{\frac{|k|}{T}}-1} \big( \widehat{b(f)}(k,\theta,\varphi), \widehat{b(h)}(-k,\theta,\varphi) \big)_{\mathscr{I}}\, dk\, d\mathbb{S}^2(\theta,\varphi),
  \end{equation*}
  where $b$ is the bulk-to-boundary projection introduced in theorem \ref{main}.
  Since the prefactor $|k| / (e^{|k| / T} - 1)$ is bounded and both $b(f)$ and $b(h)$ are square-integrable on $\mathscr{I}^+$ as per theorem \ref{main}, $\Delta_T(f,h)$ is a well-defined bidistribution.
  Moreover the prefactor decays faster than any power of $|k|$ and thus, if $(x, x', k, k')$ is a point in $\mathrm{WF}(\Delta_T)$, $k$ and $k'$ must have vanishing component along causal directions.
  As $\Delta_T$ is also a bisolution of the dynamics \eqref{dynamics}, the propagation of singularities theorem already implies that $\Delta_T$ has vanishing wavefront set, thus it must be smooth.
  Therefore $\omega^M_T$ is a Hadamard state.
\end{proof}

\section{Conclusions}\label{sec:conclusions}

The main goal of this paper is the explicit construction of physically sensible quantum states for the vector potential on a large class of curved backgrounds.
Up to now, this has been a rather elusive task on account of the gauge invariance of the underlying physical system which is known to create problems already in Minkowski spacetime.

As a first step we developed an algebraic quantization scheme for the vector potential on globally hyperbolic spacetimes following \cite{dimock:1992ea,fewster:2003or,pfenning:2009ti}.
The underlying classical dynamics is ruled by a normally hyperbolic operator once we work with gauge equivalence classes.
Thus, in principle, we could follow the standard procedure for bosonic field theories which calls for constructing a unital $*$-algebra, the field algebra, associated to the space of solutions of the equations of motion.
To this avail, one has to prove that such a space can be endowed with a weakly non-degenerate symplectic form.
Extending the result of \cite{dappiaggi:2011rt}, we proved the existence of a well-defined symplectic form for the vector potential if either the first or the second de Rham cohomology group of the underlying background is trivial.
Nevertheless, if these topological obstructions are not present, we have also shown that the vector potential can be interpreted as a locally covariant conformal quantum field theory in the sense of \cite{pinamonti:2009zn}.
Within this respect, one might wonder whether the failure of general local covariance in the general scenario could be ascribed to the fact that electromagnetism should be best understood as a particular case of a Yang-Mills theory.
In this case one would read Maxwell's equations as a theory of connections on a principal $U(1)$-bundle and the definition of a generally locally covariant quantum field theory would have to be modified.
In particular the category of globally hyperbolic spacetimes would have to be replaced by that of principal $U(1)$-bundles.
This line of reasoning will be pursued in the following papers \cite{BDS, Sanders:2012sf}.

In the second step of our analysis we explicitly construct an algebraic state of Hadamard form by means of the bulk-to-boundary reconstruction technique on asymptotically flat spacetimes first introduced in \cite{dappiaggi:2006cg}.
This method calls for identifying the bulk algebra of observables as a $*$-subalgebra of a second algebra, intrinsically built on null infinity.
The advantage of this point of view is that each state on $\mathscr{I}^+$ automatically identifies a bulk counterpart.
Since $\mathscr{I}^+$ is a three-dimensional null differentiable manifold endowed with an infinite dimensional symmetry group, the BMS group, it is possible to identify a distinguished state $\omega^{\mathscr{I}}$ on the boundary algebra.
The bulk counterpart $\omega^M$ turns out to enjoy several interesting properties.
Most notably we proved that it is invariant under the action of all isometries and that it is of Hadamard form.
Furthermore, along the lines of \cite{dappiaggi:2011ya}, we have shown that it is possible to slightly modify the form of $\omega^{\mathscr{I}}$ to construct a family of states $\omega^{\mathscr{I}}_T$ which fulfil an exact KMS condition on $\mathscr{I}^+$ with respect to the translations along the null direction.
Analogously to $\omega^{\mathscr{I}}$, also $\omega^{\mathscr{I}}_T$ induces for all $T \geq 0$ a bulk state $\omega^M_T$.
Although we cannot expect that $\omega^M_T$ is a KMS state unless the underlying background is static and thus possesses a timelike Killing field, it turns out that such a bulk state still enjoys the Hadamard property.
Therefore these states are natural candidates to deal on curved backgrounds with physical phenomena which admit a ``thermodynamical" interpretation.

Recall that we used a complete gauge fixing procedure to construct a positive state. In comparison with the past literature, \cite{strocchi:1967fk,strocchi:1970kx} in particular, this is the key difference.
While in the cited papers the Lorenz gauge condition is imposed at a Hilbert space level, in our approach it is implemented directly in the algebra along the lines first suggested by Dimock \cite{dimock:1992ea}.
Yet, in order to fix a state on the boundary and to ensure its positivity, we had to cancel one additional degree of freedom.
This has been achieved by exploiting in proposition \ref{gaugefix} the residual gauge freedom to impose the $\mathscr{I}$-gauge.
It is noteworthy that our procedure is fully covariant on account of the universal boundary structure and it induces in the bulk a Hadamard state which coincides with the standard one in Minkowski spacetime.
We expect no major obstacles repeating a similar bulk-to-boundary construction utilizing an indefinite metric ansatz similar to the Gupta-Bleuler formalism.
A detailed analysis of the relation with this method is beyond the scopes of this paper, although some preliminary remarks can be found in \cite{siemssen:2011fk}.
A state constructed via the Gupta-Bleuler scheme might be useful in treating interacting quantum field theories perturbatively.
Connecting the Gupta-Bleuler method with the bulk-to-boundary procedure would however require a vector bundle on $\mathscr{I}$ which is not naturally intrinsic to the null boundary, namely the pull-back bundle of the $4$-dimensional cotangent bundle of the unphysical spacetime.

In terms of future perspectives, we envisage that our results could be used at least in two different directions.
First of all, one can extend the whole construction to the same class of cosmological spacetimes studied in \cite{dappiaggi:2011ya} and then discuss from a mathematically rigorous point of view the role of the vector potential in the analysis of semiclassical Einstein's equation.
This potential project is closely related to the second one, namely the construction of the extended algebra of Wick polynomials for the vector potential and the associated computation of the trace anomaly for the associated stress energy tensor along the lines of the same analysis for Dirac fields as in \cite{dappiaggi:2009er,hack:2010bl}.

\paragraph{Acknowledgements.}
C. D. gratefully acknowledges financial support from the University of Pavia.
D. S. wishes to thank the Department of Theoretical and Nuclear Physics of the University of Pavia for the kind hospitality during the drawing up of this paper.
The stay has been supported by a PROMOS scholarship of the DAAD.
The present work is based on the results of the Diploma thesis of D. S. \cite{siemssen:2011fk}.
We are grateful to Klaus Fredenhagen, Thomas-Paul Hack, Valter Moretti, Nicola Pinamonti, Julien Qu\'eva and Jochen Zahn for useful discussions.


\end{document}